\documentclass[letterpaper, 10 pt, conference]{ieeeconf}  

\usepackage{amsmath,amsfonts,amsthm}
\usepackage{graphicx}
\usepackage[strict=true,style=english]{csquotes}
\usepackage{caption}
\usepackage{subcaption}
\usepackage{xcolor}
\usepackage{comment}
\usepackage{cite}
\usepackage{mathtools}
\usepackage{floatrow}

\theoremstyle{definition}
\newtheorem{problem}{Problem}
\newtheorem*{example*}{Example}
\newtheorem{definition}{Definition}
\newtheorem{proposition}{Proposition}
\newtheorem{remark}{Remark}

\newtheorem{case study}{Case Study}

\IEEEoverridecommandlockouts                              

\title{\LARGE \bf
Robust Multi-Agent Coordination from CaTL+ Specifications
}

\author{Wenliang Liu$^{1}$, Kevin Leahy$^{2}$, Zachary Serlin$^{2}$, Calin Belta$^{1}$ 
\thanks{This work was partially supported by the NSF under grant IIS-2024606.
DISTRIBUTION STATEMENT A. Approved for public release. Distribution is unlimited.
This material is based upon work supported by the Under Secretary of Defense for Research and Engineering under Air Force Contract No. FA8702-15-D-0001. Any opinions, findings, conclusions or recommendations expressed in this material are those of the author(s) and do not necessarily reflect the views of the Under Secretary of Defense for Research and Engineering.}
\thanks{$^{1}$Wenliang Liu and Calin Belta are with Boston University, Boston, MA, USA
        {\tt\small wliu97@bu.edu, cbelta@bu.edu}}%
\thanks{$^{2}$Kevin Leahy and Zachary Serlin are with MIT Lincoln Laboratory, Lexington MA 02421, USA
        {\tt\small kevin.leahy@ll.mit.edu, zachary.serlin@ll.mit.edu}}%
}

\begin{document}

\maketitle
\thispagestyle{empty}
\pagestyle{empty}

\begin{abstract}

We consider the problem of controlling a heterogeneous multi-agent system required to satisfy temporal logic requirements. Capability Temporal Logic (CaTL) was recently proposed to formalize such specifications for deploying a team of autonomous agents with different capabilities and cooperation requirements. In this paper, we extend CaTL to a new logic CaTL+, which is more expressive than CaTL and has semantics over a continuous workspace shared by all agents. We define two novel robustness metrics for CaTL+: the traditional robustness and the exponential robustness. The latter is sound, differentiable almost everywhere and eliminates masking, which is one of the main limitations of the traditional robustness metric. We formulate a control synthesis problem to maximize CaTL+ robustness and propose a two-step optimization method to solve this problem. Simulation results are included to illustrate the increased expressivity of CaTL+ and the efficacy of the proposed control synthesis approach.

\end{abstract}

\section{INTRODUCTION}

Due to their expressivity and similarity to natural languages, temporal logics, such as Linear Temporal Logic (LTL) \cite{baier2008principles} and Signal Temporal Logic (STL) \cite{maler2004monitoring}, have been widely used to formalize specifications for control systems. Many recent works proposed methodologies to generate control strategies for dynamical systems from such specifications. Roughly, these works can be grouped in two categories. For specifications given in LTL and fragments of LTL, abstractions and automata-based synthesis have been shown to work for certain classes of systems with low-dimensional state spaces \cite{belta2017formal}. 
For temporal logics with semantics over real-valued signals, such as STL, control synthesis can be formulated as an optimization problem, which can be solved via Mixed Integer Programming (MIP) \cite{belta2019formal} 
or gradient-based methods \cite{pant2017smooth}, 
\cite{gilpin2020smooth}. 

Planning and controlling multi-agent systems from temporal logic specifications is a challenging problem that received a lot of attention in recent years. In \cite{chen2011synthesis}, the authors used distributed formal synthesis to allocate global task specifications expressed as LTL formulas to locally interacting agents with finite dynamics. Related works \cite{schillinger2018,kantaros2020,luo2021} proposed extensions to more realistic models and scenarios. The authors of \cite{Liu2021} used STL and Spatial Temporal Reach and Escape Logic (STREL) as specification languages to capture the connectivity constraints of a multi-robot team.

The above approaches do not scale for large teams. 
Very recent work has focused on development of logics and synthesis algorithms specifically tailored for such situations. The authors of \cite{sahin2017provably} proposed Counting LTL (cLTL), which is used to define tasks for large groups of identical agents. \emph{Counting constraints} specify the number of agents that need to achieve a certain goal. As long as enough agents achieve the goal, it does not matter which specific subgroup does that. Similar ideas are used in Capability Temporal Logic (CaTL) \cite{leahy2021scalable}. The atomic unit of a CaTL formula is a \emph{task}, which specifies the number of agents with certain capability that need to reach some region and stay there for some time. Since agents can have different capabilities, CaTL is appropriate for heterogeneous teams. Unlike cLTL, CaTL is a fragment of STL, and allows for concrete time requirements.

The expressivities of both cLTL and CaTL are limited by the definition of \emph{counting constraints} and \emph{tasks}, respectively. For example, neither cLTL nor CaTL can require that ``$3$ agents eventually reach region A'' without requiring that the $3$ agents reach region A at the same time. This can be restrictive. Consider, for example, a disaster relief scenario, where a team of robots needs to deliver supplies to an affected area. We require that enough supplies be delivered, i.e., enough robots eventually reach the affected area, rather than enough robots stay in the affected area at the same time. In fact, a robot is supposed to go on to the next task after drop off the supply without waiting for the other robots. To solve this problem, an extension of cLTL, called cLTL+, was proposed in \cite{sahin2017synchronous}, where a two-layer LTL structure was defined. However, cLTL+ can not specify concrete time requirements. To address this limitation, we extend CaTL to a novel logic called CaTL+, which has a two-layer structure similar to cLTL+. 
The authors of \cite{buyukkocak2021planning} extended STL with integral predicates, which also enables asynchronous satisfaction, but it can only specify how many times a service is needed. A CaTL+ task can specify the number of agents that need to satisfy an arbitrary STL formula. 
Another related work is CensusSTL \cite{xu2016census}, which is also a two-layer STL that refers to mutually exclusive subsets of a group, rather than capabilities of agents. Also, \cite{xu2016census} focuses on inference of formulas from data, rather than control synthesis. 

CaTL has both qualitative semantics, i.e., a specification is satisfied or violated, and quantitative semantics (known as robustness), which defines how much a specification is satisfied. The robustness of CaTL is an integer representing the minimum number of agents that can be removed from (added to) a given team in order to invalidate (satisfy) the given formula. Such a robustness metric is discontinuous and cannot represent how strongly each \emph{task} is satisfied. In this paper, for the new logic CaTL+, we define two new quantitative semantics, called traditional robustness and exponential robustness. Both are continuous with respect to their inputs and can measure how strongly a \emph{task} is satisfied. A higher robustness indicates more agents reach the region and stay closer to the center for a longer time. Similar to the traditional robustness of STL \cite{donze2010robust}, the new traditional robustness only takes into account the most satisfaction or violation point, and ignores all other points, which is called \emph{masking}. The proposed exponential robustness eliminates \emph{masking} by considering all subformulas, all time points and all agents, which makes the control synthesis from CaTL+ easier and the result more robust. This new definition also includes a novel robustness for STL, which has the strongest mask-eliminating property in promise of soundness compared with existing STL robustness metrics.  

We also formulate and solve a centralized control synthesis problem from CaTL+. Control synthesis for cLTL~\cite{sahin2017provably}, cLTL+~\cite{sahin2017synchronous}, CaTL~\cite{leahy2021scalable} and STL with integral predicates~\cite{buyukkocak2021planning} are solved in a graph environment using a (mixed) Integer Linear Program (ILP), where the controls are transitions between vertices in a graph. In this paper, we consider a continuous workspace shared by all agents. Each agent has its own discrete time dynamics with continuous state and control spaces. We propose a two-step optimization: a global optimization followed by a gradient-based local optimization to obtain the controls that maximize CaTL+ robustness. 

The contribution of this paper is threefold. First, we extend CaTL to CaTL+, which is more expressive and works for continuous environments. Second, we define two kinds of robustness functions for CaTL+: the traditional and the exponential robustness. The latter is differentiable almost everywhere and eliminates masking, which makes it suitable for gradient-based optimization. This new definition also includes a novel STL robustness metric. Third, we propose a two-step optimization strategy for CaTL+ control synthesis. The resulting controls steer the system to satisfy the given specification robustly according to the simulation results. 

\section{System Model and Notation}
\label{sec:system}

Let $|\mathcal{S}|$ be the cardinality of a set $\mathcal{S}$. We use bold symbols to represent trajectories and calligraphic symbols for sets. For $z\in\mathbb R$, we define $[z]_+=\max\{0,z\}$ and $[z]_-=-[-z]_+$. Consider a team of agents labelled from a finite set $\mathcal{J}$. Let $Cap$ denote a finite set of agent capabilities. We assume that the agents operate in a continuous workspace $\mathcal{S}\subseteq \mathbb{R}^{n_s}$. 
\begin{definition}
An \emph{agent} $j\in\mathcal{J}$ is a tuple $A_j = \langle \mathcal X_j, x_{j}(0), Cap_j, \mathcal U_j, f_j, l_j \rangle$ where: $\mathcal X_j\subseteq\mathbb{R}^{n_{x,j}}$ is its state space; $x_{j}(0)\in\mathcal X_j$ is its initial state; $Cap_j\subseteq Cap$ is its finite set of capabilities; $\mathcal U_j\subseteq\mathbb{R}^{n_{u,j}}$ is its control space; $f_j: \mathcal{X}_j\times\mathcal{U}_j\rightarrow\mathcal{X}_j$ is a differentiable function giving the discrete time dynamics of agent $j$:
\begin{equation}
    \label{eq:system}
    x_{j}(t+1) = f_j\big(x_{j}(t),u_{j}(t)\big),\quad t=0,1,\ldots,H-1,
\end{equation}
where $x_{j}(t)$ and $u_{j}(t)$ are the state and control at time $t$, $H$ is a finite time horizon determined by the task (detailed later); $l_j: \mathcal X_j\rightarrow \mathcal S$ is a differentiable function that maps the state of agent $j$ to a point in the workspace shared by all agents (this enables heterogeneous state spaces): 
\vspace{-2pt}
\begin{equation}
\label{eq:ws}
    s_j(t) = l_j\big(x_j(t)\big),\quad t=0,1,\ldots,H.
    \vspace{-2pt}
\end{equation}
The trajectory of an agent $j$, called an \emph{individual trajectory}, is a sequence $\mathbf{s}_j=s_{j}(0) \ldots s_{j}(H)$. We assume that $\cup_{j\in \mathcal{J}}Cap_j= Cap$. 
\end{definition}

\begin{figure}
\floatbox[{\capbeside\thisfloatsetup{capbesideposition={right,top},capbesidewidth=4.5cm}}]{figure}[\FBwidth]
  {\includegraphics[width=3.5cm]{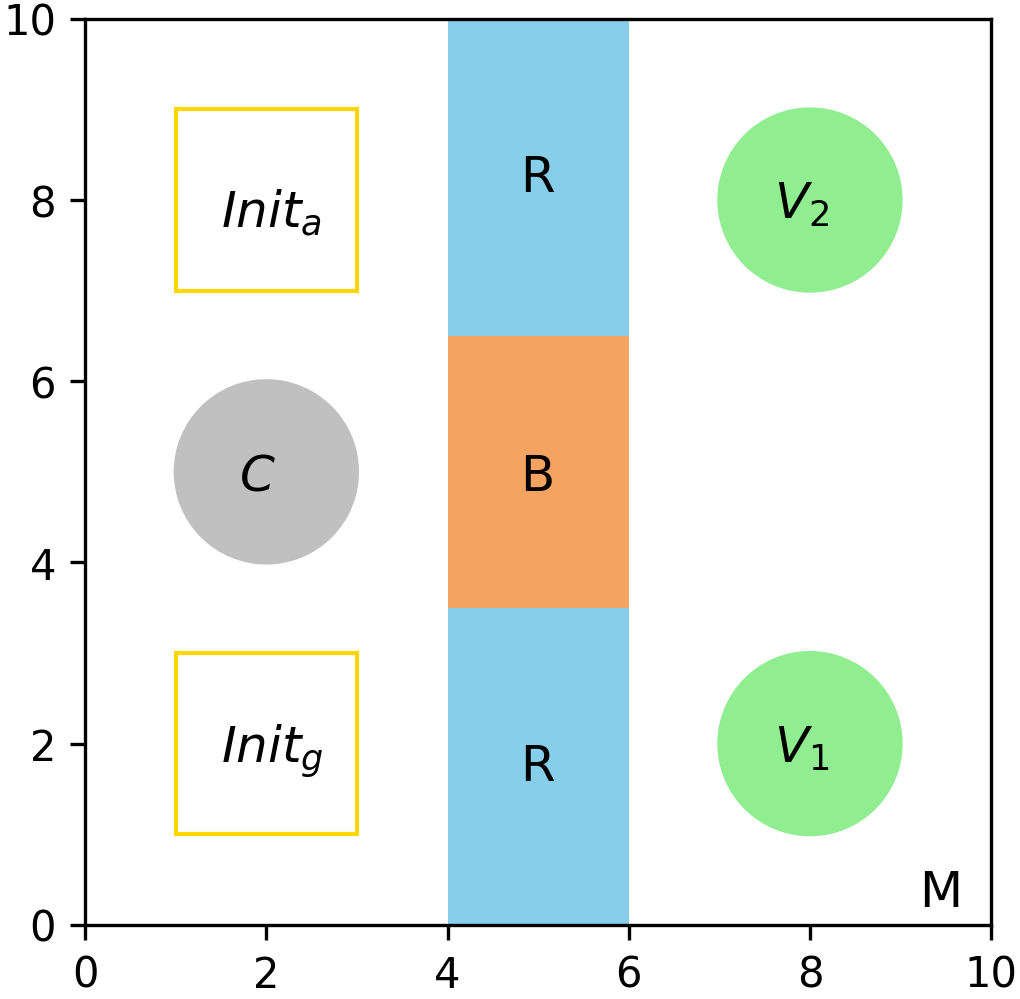}}
  {\caption{\small An earthquake emergency response scenario. $M$ is the entire square workspace. $Init_a$ and $Init_g$ are the initial regions for the aerial and ground vehicles. $C$ is where the agents pick up supplies. $B$ is a bridge and the $2$ rectangles $R$ correspond to the river. $V_1$ and $V_2$ are $2$ affected villages.}\label{fig:map}}
  \vspace{-10pt}
\end{figure}

Given a team of agents $\{A_j\}_{j\in\mathcal{J}}$, the \emph{team trajectory} is defined as a set of pairs $\mathbf S = \{(\mathbf s_j,Cap_j)\}_{j\in\mathcal J}$, which captures all the \emph{individual trajectories} with the corresponding capabilities. Let $\mathcal J_c = \{j\;|\;c\in Cap_j\}$ be the set of agent indices with capability $c$. Let $\mathbf u_j = u_{j}(0)\ldots u_{j}(H-1)$ be the sequence of controls for agent $j$.

\begin{example*}
\label{eg:1}
Consider an earthquake emergency response scenario. The workspace $\mathcal S \subset \mathbb R^2$ is shown in Fig. \ref{fig:map}. There are $4$ ground vehicles $j\in\{1,2,3,4\}$ and $2$ aerial vehicles $j\in\{5,6\}$, totaling $6$ robots indexed from $\mathcal{J} = \{1,2,3,4,5,6\}$ in the workspace. A river $R$ runs through this area and a bridge $B$ goes across the river. All ground vehicles are identical. The dynamics $f_j$, $j\in\{1,2,3,4\}$ are given by
\vspace{-1pt}
\begin{equation}
    \label{eq:system1}
    \begin{aligned}
    p_{x,j}(t+1) &=p_{x,j}(t)+v_j(t)\cos\theta_j(t),\\
    p_{y,j}(t+1) &=p_{y,j}(t)+v_j(t)\sin\theta_j(t),\\
    \theta_j(t+1) &=\theta_j(t)+\omega_j(t),
    \end{aligned}
    \vspace{-1pt}
\end{equation}
where the state $x_j$ is the 2D position and orientation $[p_{x,j}\ p_{y,j}\ \theta_j]$, the control $u_j$ is the forward and angular speed $[v_j\ \omega_j]$, the state space $\mathcal X_j = \mathcal X_g\subset\mathbb R^3$, the control space $\mathcal U_j = \mathcal U_g\subset\mathbb R^2$,
the initial state $x_j(0)$ is a singleton randomly sampled in region $Init_g$ with randomly sampled orientation $\theta_j\in[\frac 1 4 \pi, \frac 3 4 \pi]$. The function $l_j(x_j) = [p_{x,j}\ p_{y,j}] = s_j$ maps the state of agent $j$ to a position in the workspace $\mathcal S$. The identical capabilities are given by $Cap_j = \{``Delivery", ``Ground"\}$, $j\in\{1,2,3,4\}$. All the aerial vehicles are identical. For $j\in\{5,6\}$, $f_j$ are given by
\vspace{-1pt}
\begin{equation}
    \label{eq:system2}
    \begin{aligned}
    p_{x,j}(t+1) &=p_{x,j}(t)+v_{x,j}(t),\\
    p_{y,j}(t+1) &=p_{y,j}(t)+v_{y,j}(t),
    \end{aligned}
    \vspace{-1pt}
\end{equation}
where the state $x_j$ is the 2D position $[p_{x,j}\ p_{y,j}]$, the control $u_j$ is the speed $[v_{x,j}\ v_{y,j}]$, the state space $\mathcal X_j = \mathcal X_a\subset\mathbb R^2$, the control space $\mathcal U_j = \mathcal U_a\subset\mathbb R^2$,
the initial state $x_j(0)$ is a singleton  randomly sampled in the region $Init_a$, the identity mapping $l_j(x_j) = x_j = s_j$ maps the state of agent $j$ to a position in $\mathcal S$. The identical capabilities are given by $Cap_j = \{``Delivery", ``Inspection"\}$, $j\in\{5,6\}$. For this scenario, we assume the following set of requirements: (1) $6$ agents with capability $``Delivery"$ should pick up supplies from region $C$ within $8$ time units; (2) $3$ agents with capability $``Delivery"$ should deliver supplies to the affected village $V_1$ within $25$ time units, and $3$ agents with capability $``Delivery"$ should deliver supplies to the affected village $V_2$ within $25$ time units; (3) the bridge might be affected by the earthquake so any agent with capability $``Ground"$ cannot go over it until $2$ agents with capability $``Inspection"$ inspect it within $5$ time units; (4) agents with capability $``Ground"$ should always avoid entering the river $R$; (5) Since the load of the bridge is limited, at all times no more than $1$ agent with capability $``Ground"$ can be on $B$; (6) $6$ agents with capability $``Delivery"$ should always stay in region $M$.
 
\end{example*}

\section{CaTL+ Syntax and Semantics}
\label{sec:CaTL+}

In this section we introduce \emph{Capability Temporal Logic plus} (CaTL+), a logic for specifying requirements for multi-agent systems. CaTL+ has two layers: the \emph{inner logic} and the \emph{outer logic}. The inner logic is identical to STL \cite{maler2004monitoring} and is defined on an \emph{individual trajectory}. In the previous example, the inner logic can specify ``eventually visit $V_1$ within 25 time units".  The outer logic specifies behaviors of the \emph{team trajectory}. With reference to the same example, the outer logic could specify ``$3$ agents with capability ``$Delivery$" should eventually visit $V_1$ within 25 time units". 

\subsection{Inner Logic}
\begin{definition}[Syntax \cite{maler2004monitoring}]
Given an individual trajectory $\mathbf s$~\footnote{For simplicity we drop the subscript $j$ from $\mathbf s$.}, the syntax of the inner logic can be recursively defined as:
\vspace{-2pt}
\begin{equation}
\label{eq:syntax-in}
\varphi:=True\;|\;\mu \; | \; \neg\varphi \; | \; \varphi_1\land\varphi_2 \; | \; \varphi_1\lor\varphi_2 \;|  \;  \varphi_1 \mathbf{U}_{[a,b]} \varphi_2,
\vspace{-2pt}
\end{equation}
where $\varphi$, $\varphi_1$ and $\varphi_2$ are inner logic (STL) formulas, $\mu$ is a predicate in the form $h\big(s(t)\big)\geq0$. We assume $h:\mathcal S\rightarrow \mathbb R$ is differentiable in this paper. $\neg$, $\land$, $\lor$ are the Boolean operators \emph{negation}, \emph{conjunction} and \emph{disjunction} respectively. $U_{[a,b]}$ is the temporal operator \emph{Until}, where $[a,b]$ is the time interval containing all integers between $a$ and $b$ with $a,b\in\mathbb Z_{\geq 0}$. The temporal operators \emph{Eventually} and \emph{Always} can be defined as $\mathbf F_{[a,b]}\varphi = True\mathbf{U}_{[a,b]}\varphi$ and $\mathbf G_{[a,b]}\varphi = \neg\mathbf F_{[a,b]}\neg\varphi$. 
\end{definition}

The qualitative semantics of the inner logic, i.e., whether a formula $\varphi$ is satisfied by an individual trajectory $\mathbf s$ at time $t$ (denoted by $(\mathbf s,t) \models \varphi$) is same as STL \cite{maler2004monitoring}. In plain English, $\varphi_1\mathbf{U}_{[a,b]}\varphi_2$ means ``$\varphi_2$ must become \emph{True} at some time point in $[a,b]$ and $\varphi_1$ must be $True$ at all time points before that''. $\mathbf{F}_{[a,b]}\varphi$ states ``$\varphi$ is \emph{True} at some time point in $[a,b]$'' and  $\mathbf{G}_{[a,b]}\varphi$ states ``$\varphi$ must be \emph{True} at all time points in $[a,b]$''. 

\subsection{Outer Logic}
 The basic component of the outer logic, which with a slight abuse of terminology we will refer to as CaTL+, is a \emph{task} $T=\langle\varphi,c,m\rangle$, where $\varphi$ is an inner logic formula, $c\in Cap$ is a capability, and $m$ is a positive integer. The syntax of CaTL+ is defined over a \emph{team trajectory} $\mathbf S$ as:
\vspace{-1pt}
\begin{equation}
\label{eq:syntax-out}
\Phi:=True\;|\;T \; | \; \neg\Phi \; | \; \Phi_1\land\Phi_2 \; | \; \Phi_1\lor\Phi_2 \;|  \;  \Phi_1 \mathbf{U}_{[a,b]} \Phi_2,
\vspace{-2pt}
\end{equation}
where $\Phi$, $\Phi_1$ and $\Phi_2$ are CaTL+ formulas, $T=\langle\varphi,c,m\rangle$ is a task, and the other operators are the same as in the inner logic defined above. Before defining the qualitative semantics of CaTL+, we introduce a counting function $n(\mathbf S,c,\varphi,t)$:
\vspace{-2pt}
\begin{equation}
    n(\mathbf S,c,\varphi,t) = \sum_{j\in\mathcal J_c}I\big( (\mathbf s_j,t)\models \varphi \big),
    \vspace{-2pt}
\end{equation}
where $I$ is an indicator function, i.e., $I = 1$ if $(\mathbf s_j,t)\models \varphi$ and $I = 0$ otherwise. $n(\mathbf S,c,\varphi,t)$ captures how many \emph{individual trajectories} $\mathbf s_j$ with capability $c$ in the \emph{team trajectory} $\mathbf S$ satisfy an inner logic formula $\varphi$ at time $t$. The qualitative semantics of the \emph{outer logic} is similar to the one of the \emph{inner logic}, except for it involves \emph{tasks} rather than predicates.

\begin{definition}
A \emph{team trajectory} $\mathbf S$ satisfies a \emph{task} $T=\langle\varphi, c, m\rangle$ at $t$, denoted by $(\mathbf S, t)\models T$, iff $n(\mathbf S,c,\varphi,t) \geq m$.
\end{definition}

In words, a task $T=\langle\varphi,c,m\rangle$ is satisfied at time $t$ if and only if at least $m$ \emph{individual trajectories} of agents with capability $c$ satisfy $\varphi$ at time $t$. The semantics for the other operators are identical with the ones in the \emph{inner logic}. We denote the satisfaction of a CaTL+ formula $\Phi$ at time $t$ by a \emph{team trajectory} $\mathbf S$ as $(\mathbf S,t)\models\Phi$. Note that specifying no more than $m$ agents with capability $c$ that could satisfy $\varphi$ can be formulated as $\Phi=\neg\langle\varphi,c,m+1\rangle$. Let the \emph{time horizon} of a CaTL+ formula $\Phi$, denoted by $hrz(\Phi)$, be the closest time point in the future that is needed to determine the satisfaction of $\Phi$. Note that cooperative inner logic is not allowed in CaTL+ since $\varphi$ is defined for one agent only.

\begin{example*}
\label{eg:2}
(continued) Reaching a circular or rectangular region can be formulated as an inner logic formula easily. For brevity, we use $s\in B$ (other regions are the same) to represent the inner logic formula of reaching region $B$. The requirements in the previous example can be formulated as CaTL+ formulas:
(1) $\Phi_1=\langle \mathbf F_{[0,8]} s\in C,\ ``Delivery",\ 6 \rangle$;
    (2) $\Phi_2=\langle \mathbf F_{[0,25]} s\in V_1,\ ``Delivery",\ 3 \rangle \land \langle \mathbf F_{[0,25]} s\in V_2,\ ``Delivery",\ 3 \rangle$;
    (3) $\Phi_3=\neg\langle s\in B,\ ``Ground",\ 1 \rangle \mathbf  U_{[0,5]}\langle s\in B,\ ``Inspection",\ 2 \rangle$ ;
    (4) $\Phi_4=\mathbf G_{[0,25]}\langle \neg(s\in R),\ ``Ground",\ 4 \rangle$;
    (5) $\Phi_5=\mathbf G_{[0,25]}\neg\langle s\in B,\ ``Ground",\ 2 \rangle$;
    (6) $\Phi_6=\mathbf G_{[0,25]}\langle s\in M,\ ``Delivery",\ 6 \rangle$. The overall specification for the system is $\Phi = \bigwedge_{i=1}^6 \Phi_i$, with $hrz(\Phi)=25$.
\end{example*}

\subsection{CaTL+ Quantitative Semantics}
The qualitative semantics defined above provides a \emph{True} or \emph{False} value, meaning that the CaTL+ formula is satisfied or not. In this section, we define its quantitative semantics, also known as robustness. This is a real value that measures how much a formula is satisfied. We introduce two robustness metrics for CaTL+, the traditional robustness (inspired by \cite{donze2010robust}) and the exponential robustness.

For simplicity, we give the definition of CaTL+ robustness in a structured manner. We show that a robustness metric for CaTL+ can be captured by only the robustness for conjunction and \emph{task}. According to the De Morgan law, disjunction can be replaced by conjunction and negation, i.e., $\Phi_1\lor\Phi_2=\neg(\neg\Phi_1\land\neg\Phi_2)$. Meanwhile, the temporal operator ``\emph{always}" and ``\emph{eventually}" can be regarded as conjunction and disjunction evaluated over individual time steps. For any robustness metric, the robustness of a predicate $h\big(s(t)\big)\geq0$ is $h\big(s(t)\big)$ and the robustness of $\neg\Phi$ is the negative of the robustness of $\Phi$ (see \cite{varnai2020robustness} for details).

\subsubsection{Traditional Robustness}
\label{sec:trad_ro}

With a slight abuse of notations, we denote the traditional robustness of an inner logic (STL) formula $\varphi$ over an \emph{individual trajectory} $\mathbf s$ and an outer logic (CaTL+) formula $\Phi$ over a \emph{team trajectory} $\mathbf S$ at time $t$ as $\rho(\mathbf s, \varphi, t)$ and $\rho(\mathbf S, \Phi, t)$, respectively. The robustness of the conjunction over $M$ subformulas with robustness values $\rho_1,\ldots,\rho_M$ is defined as \cite{donze2010robust}:
\vspace{-2pt}
\begin{equation}
    \label{eq:traditional_and}
    \mathcal A^{trad}(\rho_1,\ldots,\rho_M) =\min(\rho_1,\ldots,\rho_M).
    \vspace{-2pt}
\end{equation}
To define the robustness of a \emph{task}, we introduce a function $L_m:\mathbb R^{n_c}\rightarrow \mathbb R$ that finds the $m^{th}$ largest element in a vector with $n_c$ elements, where $m\leq n_c$. Then the traditional robustness for a \emph{task} $T=\langle \varphi,c,m\rangle$ over a \emph{team trajectory} $\mathbf S$ at time $t$ is defined as:
\vspace{-2pt}
\begin{equation}
    \label{eq:traditional_task}
    \rho(\mathbf S,\langle \varphi,c,m\rangle,t) = L_m([\rho(\mathbf s_j, \varphi,t)]_{j\in\mathcal J_c}),
\end{equation}
where $[\rho(\mathbf s_j, \varphi,t)]_{j\in\mathcal J_c}\in\mathbb R^{|\mathcal J_c|}$ is a vector containing all $\rho(\mathbf s_j, \varphi,t)$ with $j\in\mathcal J_c$. That is, the traditional robustness of a \emph{task} $T=\langle \varphi,c,m\rangle$ is the $m^{th}$ largest robustness value of $\varphi$ over all \emph{individual trajectories} with capability $c$. The greater this value, the stronger the satisfaction of the task $T$ will be. Then the traditional robustness of CaTL+ can be recursively constructed from \eqref{eq:traditional_and} and \eqref{eq:traditional_task}. 

An important property of a robustness metric is soundness, i.e., positive robustness indicates (qualitative) satisfaction of the formula, while negative robustness corresponds to violation. Formally, we have:

\begin{definition}[Soundness]
A robustness metric $\rho(\mathbf S,\Phi,t)$ is \emph{sound} if for any formula $\Phi$, $\rho(\mathbf S,\Phi,t)\geq 0$ iff $(\mathbf S,t)\models\Phi$.
\end{definition}
\begin{proposition}
\label{pp:sound-trad}
Traditional robustness for CaTL+ is sound.
\end{proposition}
\begin{proof}
Consider a \emph{task} $T=\langle \varphi,c,m\rangle$. Obviously, $L_m([\rho(\mathbf s_j, \varphi,t)]_{j\in\mathcal J_c})\geq0$ iff at least $m$ entries in $[\rho(\mathbf s_j, \varphi,t)]_{j\in\mathcal J_c}$ are non-negative, which means that $(\mathbf S,t)\models T$. Since all the other operators are identical with STL \cite{donze2010robust}, the traditional robustness for CaTL+ is sound.
\end{proof}
\vspace{-5pt}


Next, we introduce another robustness metric for CaTL+, called exponential robustness. Using exponential robustness makes the optimization process for control synthesis easier. We will compare it with traditional robustness and discuss its advantages after we formally define it. 

\subsubsection{Exponential robustness}
\label{sec:exp_ro}
We use $\eta(\mathbf s, \varphi, t)$ and $\eta(\mathbf S, \Phi, t)$ to denote the exponential robustness of the inner and outer logic. As discussed earlier, we only need to define the robustness for conjunction and \emph{task}. Consider the conjunction over $M$ subformulas with robustness $\eta_1,\ldots,\eta_M$. Similar to \cite{varnai2020robustness}, we first define an effective robustness measure, denoted by $\eta_i^{conj}$, $i=1,\ldots,M$, for each subformula:
\begin{equation}
\label{eq:exp-conj-eff}
    \eta_{i}^{conj} \coloneqq \left\{
    \begin{aligned}
    &\eta_{min}e^{\frac{\eta_i-\eta_{min}}{\eta_{min}}} \quad & \eta_{min} < 0 \\
    &\eta_{min}(2 - e^{\frac{\eta_{min}-\eta_i}{\eta_{min}}}) \quad & \eta_{min} > 0 \\
    & 0 \quad & \eta_{min} = 0
    \end{aligned}\right.
\end{equation}
where $\eta_{min}=\min(\eta_1,\ldots,\eta_M)$. 
The relation between $\eta_i^{conj}$ and $\eta_i$ is shown in Fig. \ref{fig:exp-and+} and Fig.\ref{fig:exp-and-}. Intuitively, \eqref{eq:exp-conj-eff} ensures that $\eta_i^{conj}$ has the same sign with $\eta_{min}$, $\forall i=1,\ldots,M$, and $\eta_i^{conj}$ increases monotonically with $\eta_i$. Note that $\eta_i^{conj}=\eta_{i}$ when $\eta_i=\eta_{min}$.  We define the exponential robustness for conjunction as \eqref{eq:exp-and} which ensures soundness:
\vspace{-2pt}
\begin{equation}
\label{eq:exp-and}
    \mathcal A^{exp}(\eta_1,\ldots,\eta_M) = \beta\eta_{min} +  (1-\beta)\frac{1}{M}\sum_{i=1}^M \eta_i^{conj},
    \vspace{-2pt}
\end{equation}
where $\beta\in[0,1]$ balances the contribution between $\eta_{min}$ and the mean of $\eta_i^{conj}$ (same sign as $\eta_{min}$). The exponential robustness turns to be the traditional robustness when $\beta=1$. 

Now consider a \emph{task} $T=\langle \varphi,c,m\rangle$. For brevity, let $\eta_j=\eta(\mathbf s_j, \varphi, t)$ when $\varphi$ and $t$ are clear from the context. We reorder $\{\eta_j\}_{j\in\mathcal J_c}$ from the largest to the smallest, i.e., $\eta_{j_1}\geq\ldots\geq\eta_{j_m}\geq\ldots\geq\eta_{j_n}$,
where $j_k\in\mathcal J$, $k=1,\ldots,n$,  $n=|\mathcal J_c|$. Note that $\eta_{j_m}$ is the critical $m^{th}$ largest robustness, i.e., $\eta_{j_m}=L_m([\eta(\mathbf s_j, \varphi,t)]_{j\in\mathcal J_c})$. We define another effective robustness $\eta_{j_k}^{task}$ for each $j_k$, $k=1,\ldots,|\mathcal J_c|$:
\vspace{-2pt}
\begin{equation}
    \label{eq:exp-task-eff}
    \eta_{j_k}^{task} \coloneqq \left\{
    \begin{aligned}
    &\frac{2\alpha(e^{\eta_{j_m}}-1)}{1+e^{-\alpha(\eta_{j_k}-\eta_{j_m})}} \quad & \eta_{j_m} > 0, \\
    &\frac{-2\alpha(e^{-\eta_{j_m}}-1)}{1+e^{\alpha(\eta_{j_k}-\eta_{j_m})}} \quad & \eta_{j_m} \leq 0,
    \end{aligned}\right.
    \vspace{-2pt}
\end{equation}
where $\alpha>0$. The relation between $\eta_{j_k}^{task}$ and $\eta_{j_k}$ is shown in Fig. \ref{fig:exp-task+} and Fig. \ref{fig:exp-task-} ($\alpha=1$). Similar to conjunction, \eqref{eq:exp-task-eff} ensures that $\eta_{j_k}^{task}$ has the same sign with the critical $\eta_{j_m}$, $\forall k=1,\ldots,|\mathcal J_c|$ and $\eta_{j_k}^{task}$ increases monotonically with $\eta_{j_k}$. Note that $\eta_{j_k}^{task} = sgn(\eta_{j_m})\alpha(e^{|\eta_{j_m}|}-1)$ when $k=m$. Again, we define the exponential robustness for a \emph{task} as the mean of $\eta_{j_k}^{task}$ to ensure soundness:
\vspace{-2pt}
\begin{equation}
\label{eq:exp-task}
    \eta(\mathbf S, \langle \varphi,c,m\rangle, t) = \frac{1}{|\mathcal J_c|}\sum_{k=1}^{|\mathcal J_c|} \eta_{j_k}^{task}.
    \vspace{-2pt}
\end{equation}

\begin{figure}
    \centering
    \begin{subfloat}[\small$\eta_{min}>0$\label{fig:exp-and+}]{
    \includegraphics[width=0.4\linewidth]{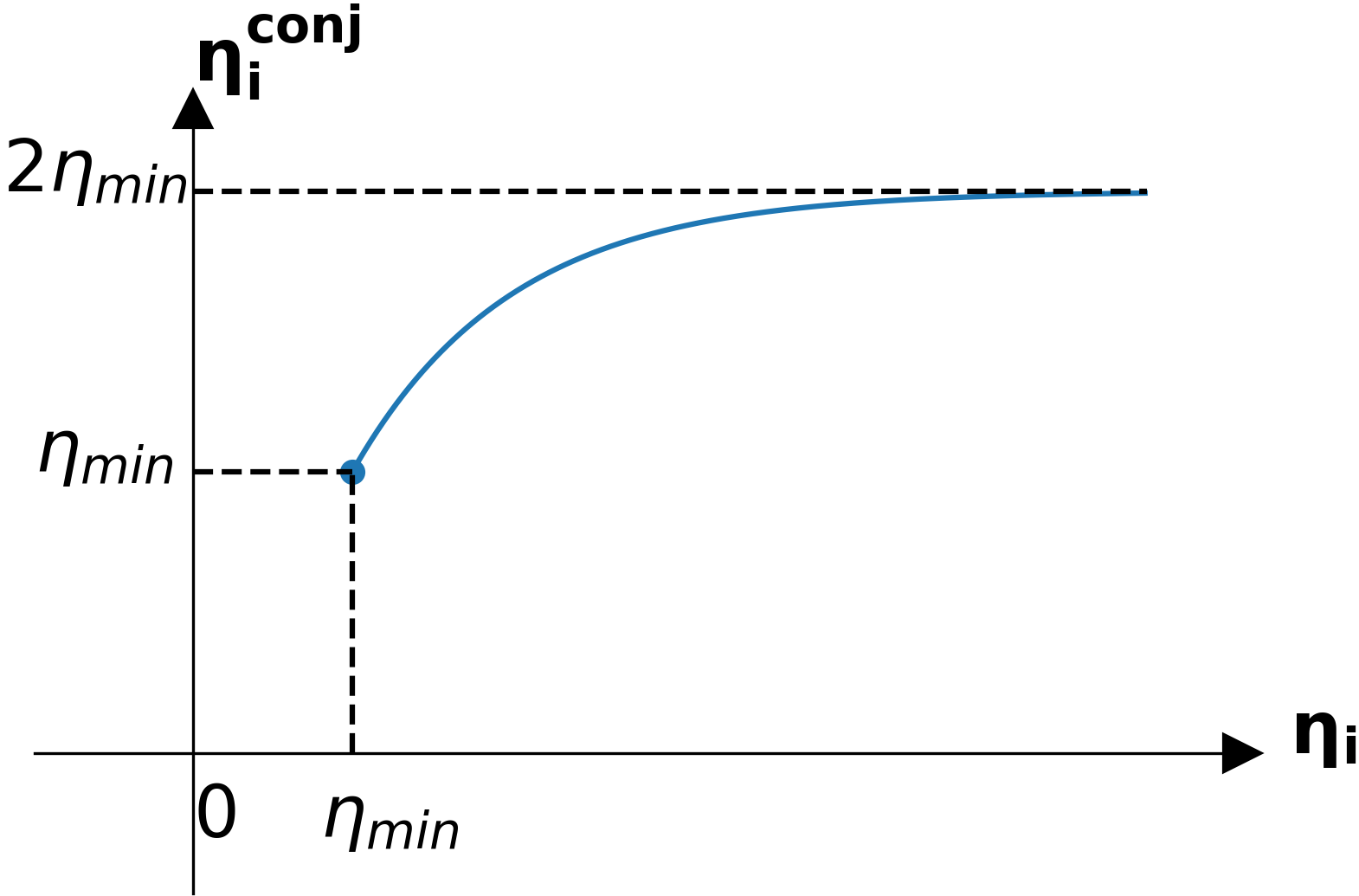}}
    \end{subfloat}
    \quad
    \begin{subfloat}[\small$\eta_{min}<0$\label{fig:exp-and-}]{
    \includegraphics[width=0.4\linewidth]{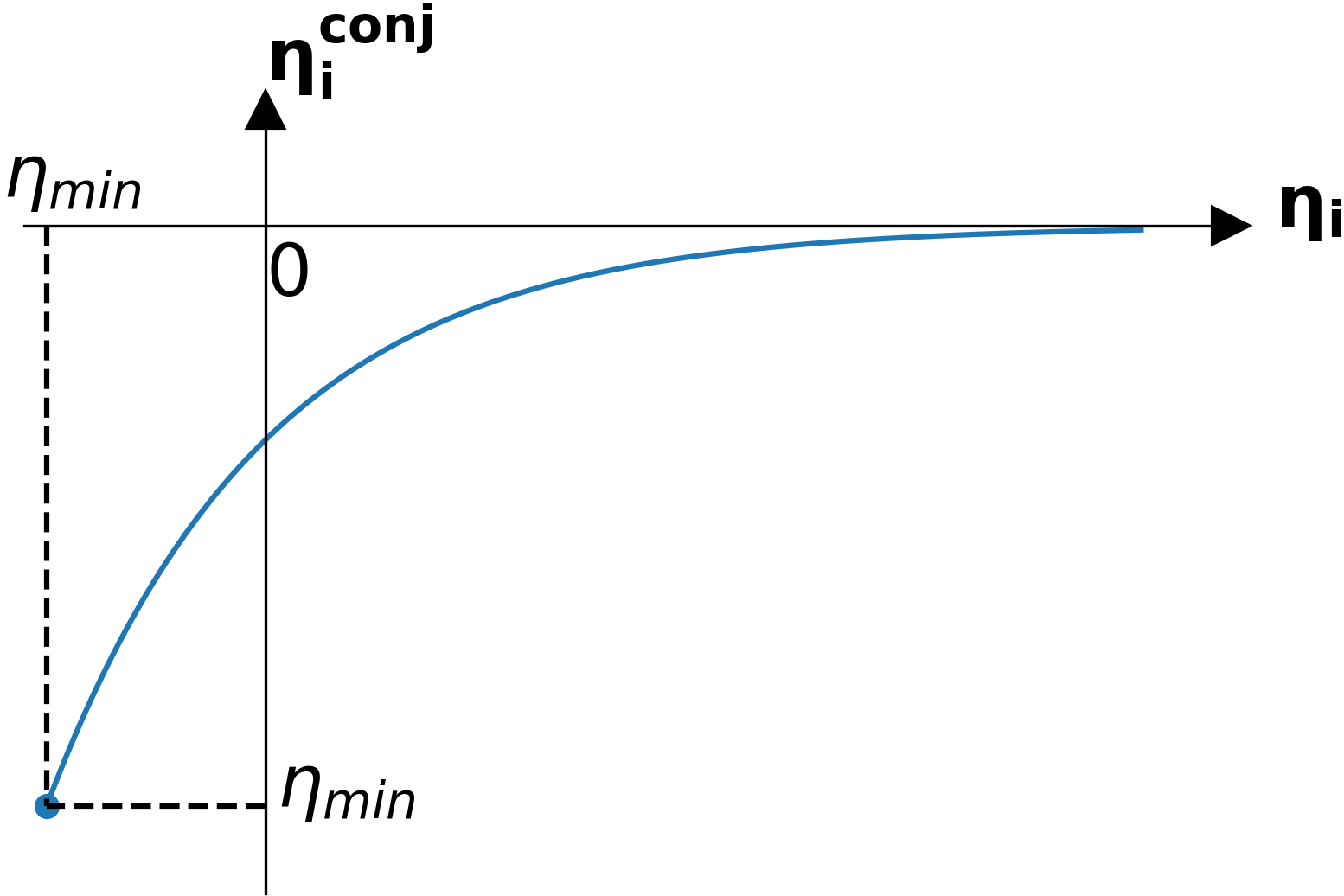}}
    \end{subfloat}
    \vspace{+5pt}
    
    \begin{subfloat}[\small$\eta_{j_m}>0$\label{fig:exp-task+}]{
    \includegraphics[width=0.4\linewidth]{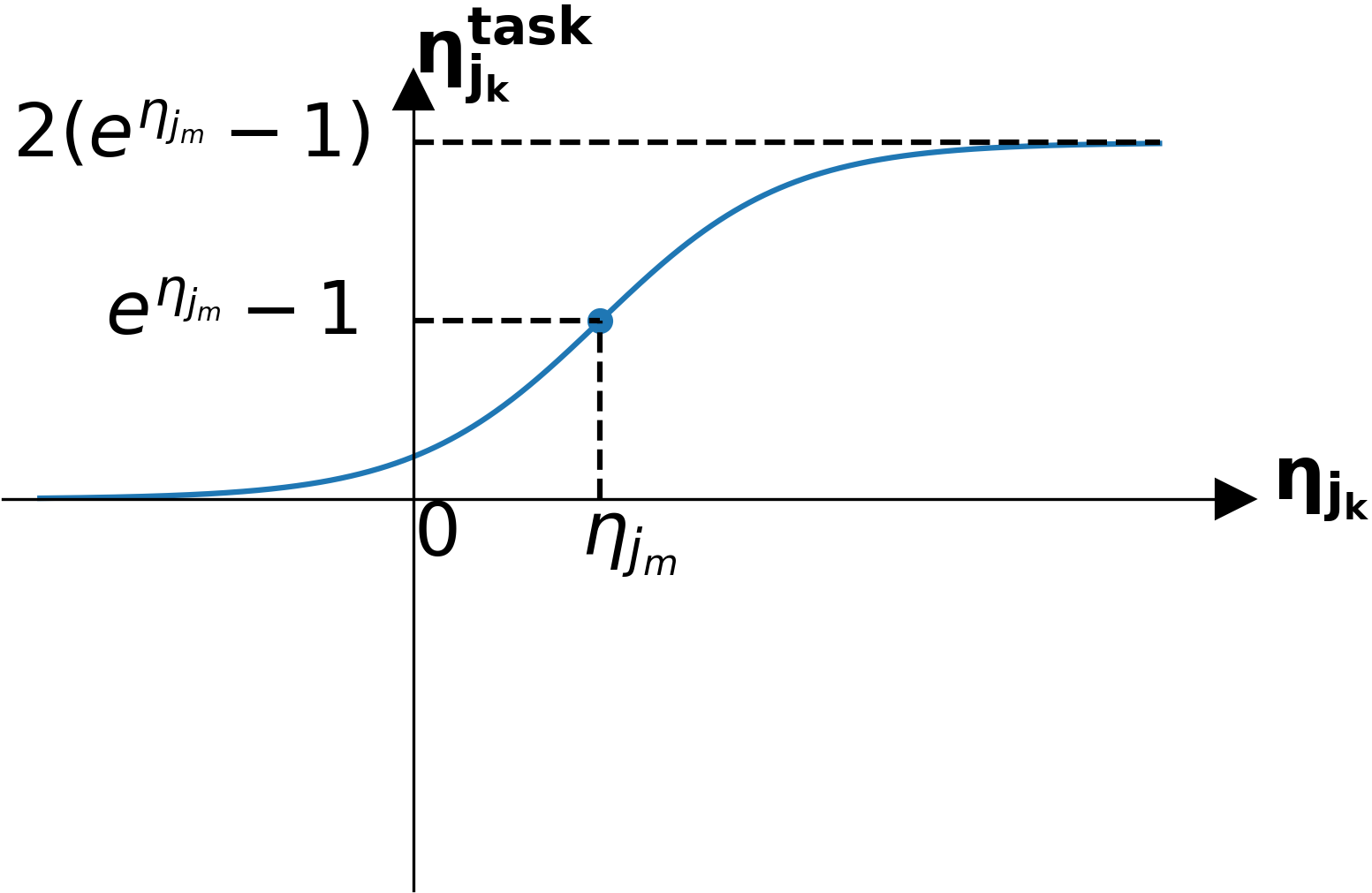}}
    \end{subfloat}
    \quad
    \begin{subfloat}[\small$\eta_{j_m}<0$\label{fig:exp-task-}]{
    \includegraphics[width=0.4\linewidth]{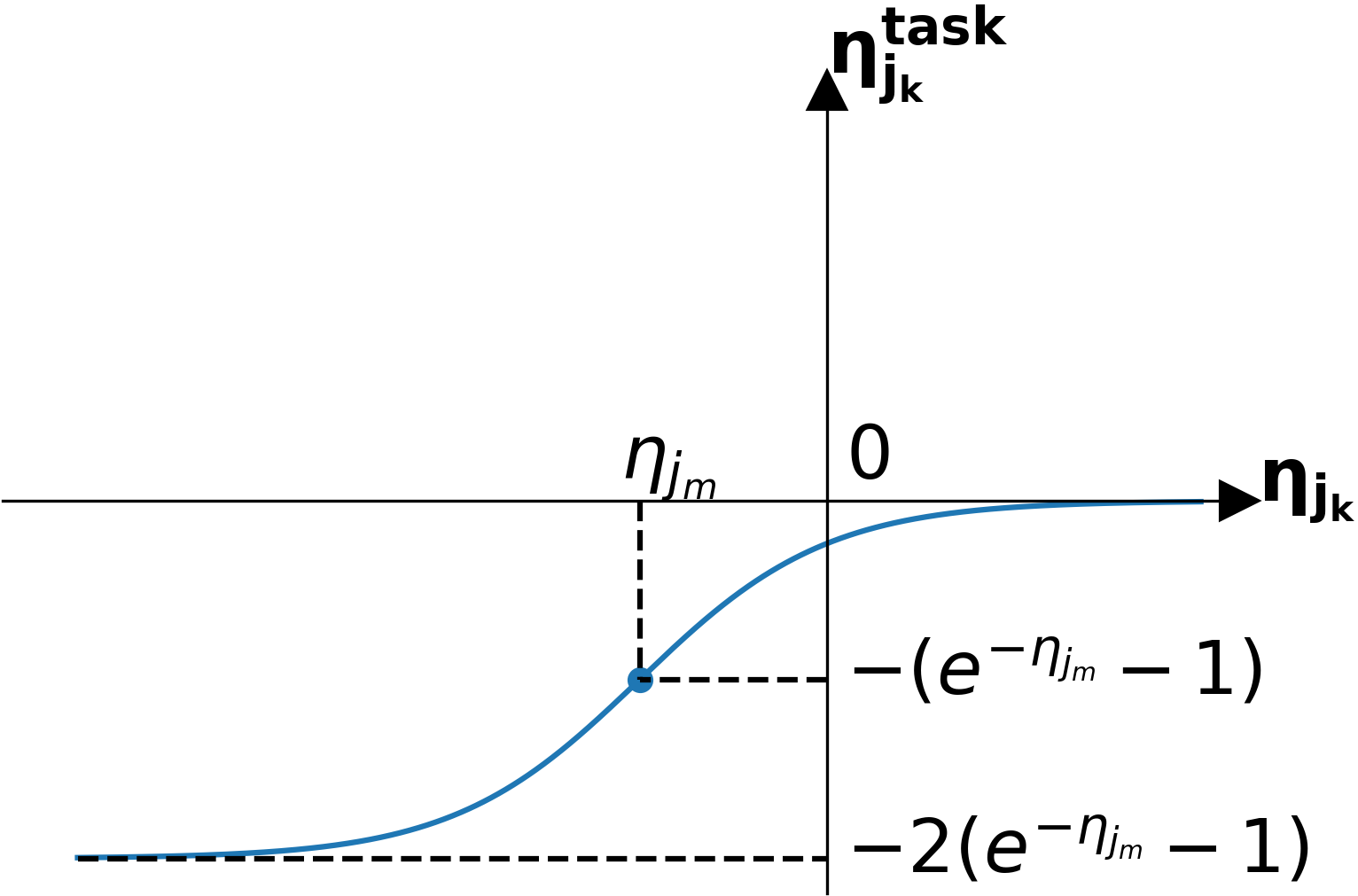}}
    \end{subfloat}
    \caption{\small (a), (b): Relation between $\eta_i^{conj}$ and $\eta_i$ while holding $\eta_{min}$ constant and $\eta_{i}\neq\eta_{min}$. (c), (d): Relation between $\eta_{j_k}^{task}$ and $\eta_{j_k}^{task}$ ($\alpha=1$) while holding $\eta_{j_m}$ constant and $k\neq m$.\vspace{-10pt}}
    \label{fig:exp}
    \vspace{-5pt}
\end{figure}

When $\alpha\rightarrow\infty$, the exponential robustness of a task only depends on $\eta_{j_m}$. The exponential robustness for CaTL+ is recursively constructed from \eqref{eq:exp-and} and \eqref{eq:exp-task}. In the following, we discuss the properties of the exponential robustness. 

\begin{proposition}
\label{pp:sound}
Exponential robustness for CaTL+ is sound.
\end{proposition}
\begin{proof}
This follows directly from the soundness of both conjunction and \emph{task}. 
\end{proof}
\vspace{-5pt}

In an optimal control problem, such as the one considered in  Sec. \ref{sec:formulation}, it is desirable to have a differentiable robustness allowing for gradient based optimization methods. 

\begin{proposition}
\label{pp:diff}
The exponential robustness $\eta(\mathbf S, \Phi, t)$ is continuous everywhere and differentiable almost everywhere with respect to individual trajectories $\mathbf s_j$, $j\in\mathcal J$.
\end{proposition}
\begin{proof}
\eqref{eq:exp-conj-eff} and \eqref{eq:exp-task-eff} are piece-wise continuous, where the switches happen at $\eta_{min}=0$ and $\eta_{j_m}=0$. For all $i$ and $k$:
\begin{equation*}
    0=\lim_{\eta_{min}\to 0^+}\eta_i^{conj} = \lim_{\eta_{min}\to 0^-}\eta_i^{conj} = \eta_i^{conj}|_{\eta_{min}=0},
\end{equation*}
\begin{equation*}
    0=\lim_{\eta_{j_m}\to 0^+}\eta_{j_k}^{task} = \lim_{\eta_{j_m}\to 0^-}\eta_{j_k}^{task} = \eta_{j_k}^{task}|_{\eta_{j_m}=0}.
\end{equation*}
Hence, $\eta_i^{conj}$ and $\eta_{j_k}^{task}$ are continuous with respect to $\eta_i$ and $\eta_{j_k}$. Since \eqref{eq:exp-and}, \eqref{eq:exp-task} and $h$ are differentiable, the continuity property is proved. When $\eta_{min}$ and $\eta_{j_m}$ are unique and nonzero, $\eta_i^{conj}$ and $\eta_{j_k}^{task}$ are differentiable with respect to $\eta_i$ and $\eta_{j_k}$. The set of non-differentiable points has measure $0$. Hence, $\eta(\mathbf S, \Phi, t)$ is differentiable almost everywhere.
\end{proof}

Although exponential robustness is not differentiable everywhere, in a numerical optimization process it rarely gets to the non-differentiable points. Moreover, these points are semi-differentiable, so even if they are met, we can still use the left or right derivative to keep the optimization running. 

The traditional robustness of CaTL+ is also sound (Proposition \ref{pp:sound-trad}). It is easy to prove that the traditional robustness is continuous everywhere and differentiable almost everywhere. The most important advantage of the exponential robustness over traditional robustness is that it eliminates \emph{masking}. In short, the traditional robustness only takes into account the minimum robustness in conjunction and the $m^{th}$ largest robustness in a \emph{task}. All the other subformulas and \emph{individual trajectories} have no contribution to the overall robustness. For example, consider formula $\mathbf F_{[0,5]} (s>3)$. Two trajectories $1,1,1,1,5$ and $1,2,3,4,5$ get the same traditional robustness score of $2$, though it is obvious that the later is more robust under disturbances.  Similarly, for a \emph{task} $\langle\varphi,c,3\rangle$, $3$ agents satisfying $\varphi$ and $4$ agents satisfying $\varphi$ to the same extent obtain same  traditional robustness, but the later is more robust to agent attrition. The exponential robustness addresses this \emph{masking} problem by taking into account the robustness for all subformulas, all time points, and all agents. As a result, the exponential robustness rewards the trajectories that satisfy the requirements at more time steps and with more agents. From an optimization point of view, our goal is to synthesize controls that maximize the CaTL+ robustness. If we use the traditional robustness, at each optimization step we can only modify the most satisfying or violating points. This may decelerate the optimization speed or even result in divergence.  
The exponential robustness makes the robustness-based control synthesis easier, and makes the results more robust. Formally, we have:

\begin{definition}\label{def:mask}[mask-eliminating]
The robustness of an operator $\mathcal O(\rho_1,\ldots,\rho_M)$ has the \emph{mask-eliminating} property if it is differentiable almost everywhere, and wherever it is differentiable, it satisfies:
\begin{equation}
\label{eq:shadow}
\frac{\partial\mathcal O(\rho_1,\ldots,\rho_M)}{\partial\rho_i}>0,\quad \forall i=1,\ldots,M,
\end{equation}
A robustness metric of CaTL+ has the mask-eliminating property if both conjunction and \emph{task} satisfy \eqref{eq:shadow}. 
\end{definition}


\begin{proposition}
Exponential robustness of CaTL+ has the mask-eliminating property.
\end{proposition}
\begin{proof}
$[$Sketch$]$ 
This follows from standard derivations. Note that $\partial \mathcal A^{exp} / \partial \eta_{min}$ depends on all $\eta_{i}^{conj}$, $i=1,\ldots,M$, and $\partial \eta(\mathbf S, \langle \varphi,c,m\rangle, t) / \partial \eta_{j_m}$ depends on all $\eta_{j_k}^{task}$, $k=1,\ldots,n$.
\end{proof}
\vspace{-7pt}

The \emph{mask-eliminating} property of exponential robustness tells us that the increase of any component in conjunction or \emph{task} results in the increase of the overall robustness. 
\begin{remark}
By standard derivations, it can also be proved that, $\partial \mathcal A^{exp} / \partial \eta_{min}>\partial \mathcal A^{exp} / \partial \eta_{i}$, $\forall \eta_i\neq \eta_{min}$. The further $\eta_i$ is from $\eta_{min}$, the smaller the partial derivative will be. Meanwhile, $\partial \eta(\mathbf S, \langle \varphi,c,m\rangle, t) / \partial \eta_{j_m}\geq\partial \eta(\mathbf S, \langle \varphi,c,m\rangle, t) / \partial \eta_{j_k}$, $\forall k\neq m$. The further $\eta_{j_k}$ is from $\eta_{j_m}$, the smaller the partial derivative will be. This is helpful because $\eta_{min}$ and $\eta_{j_m}$ are the most critical components which decide the satisfaction of the formula. 
\end{remark}

\begin{remark}
The exponential robustness for conjunction itself forms a novel robustness of STL, which has desired properties including soundness and mask-eliminating (the conjunction satisfies \eqref{eq:shadow}).  Other robustness metrics including Arithmetic-Geometric Mean (AGM) robustness \cite{mehdipour2019arithmetic} and learning robustness \cite{varnai2020robustness} are also sound and partially solve the \emph{masking} problem, but none of them satisfy the \emph{mask-eliminating} property (as shown in Fig.~\ref{fig:compare}).
The smooth max-min robustness from \cite{pant2017smooth,gilpin2020smooth} has the mask-eliminating property, but loses soundness (in a necessary and sufficient sense). To the best of our knowledge, exponential robustness is the first that satisfies both of these two properties. Sample behaviors of $\mathcal A^{exp}(\eta_1,\eta_2)$ for different $\eta_1$ and $\eta_2$ are depicted in Fig. \ref{fig:compare}. We can see that traditional and AGM robustness have $0$ partial derivatives and learning robustness has negative partial derivatives at some points, while smooth max-min robustness violates soundness. 
\end{remark}
\begin{figure}
    \centering
    \begin{subfloat}[\small$\eta_{1}=-1$\label{fig:compare-}]{
    \includegraphics[height=4cm]{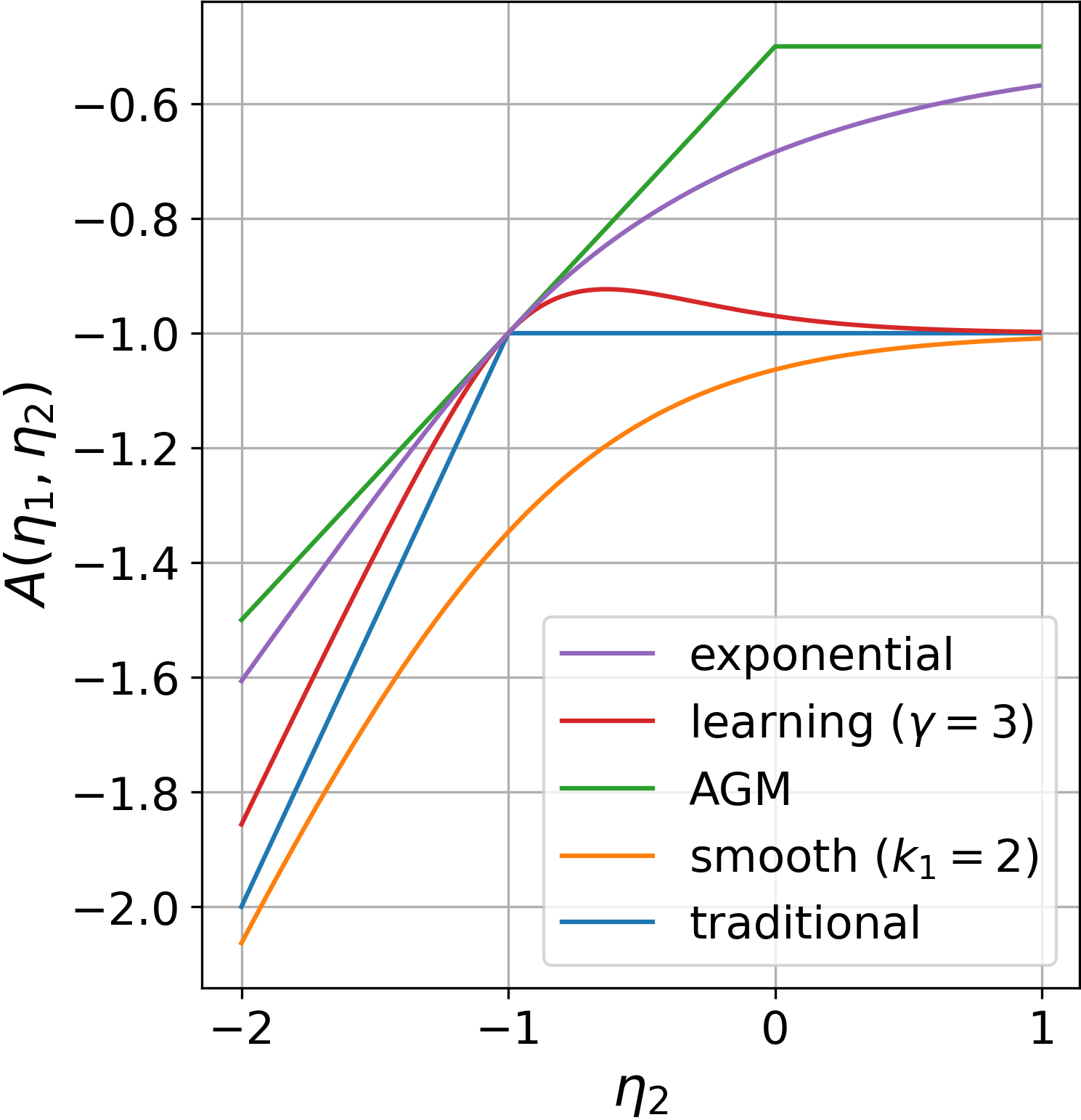}}
    \end{subfloat}
    \quad
    \begin{subfloat}[\small$\eta_{1}=1$\label{fig:compare+}]{
    \includegraphics[height=4cm]{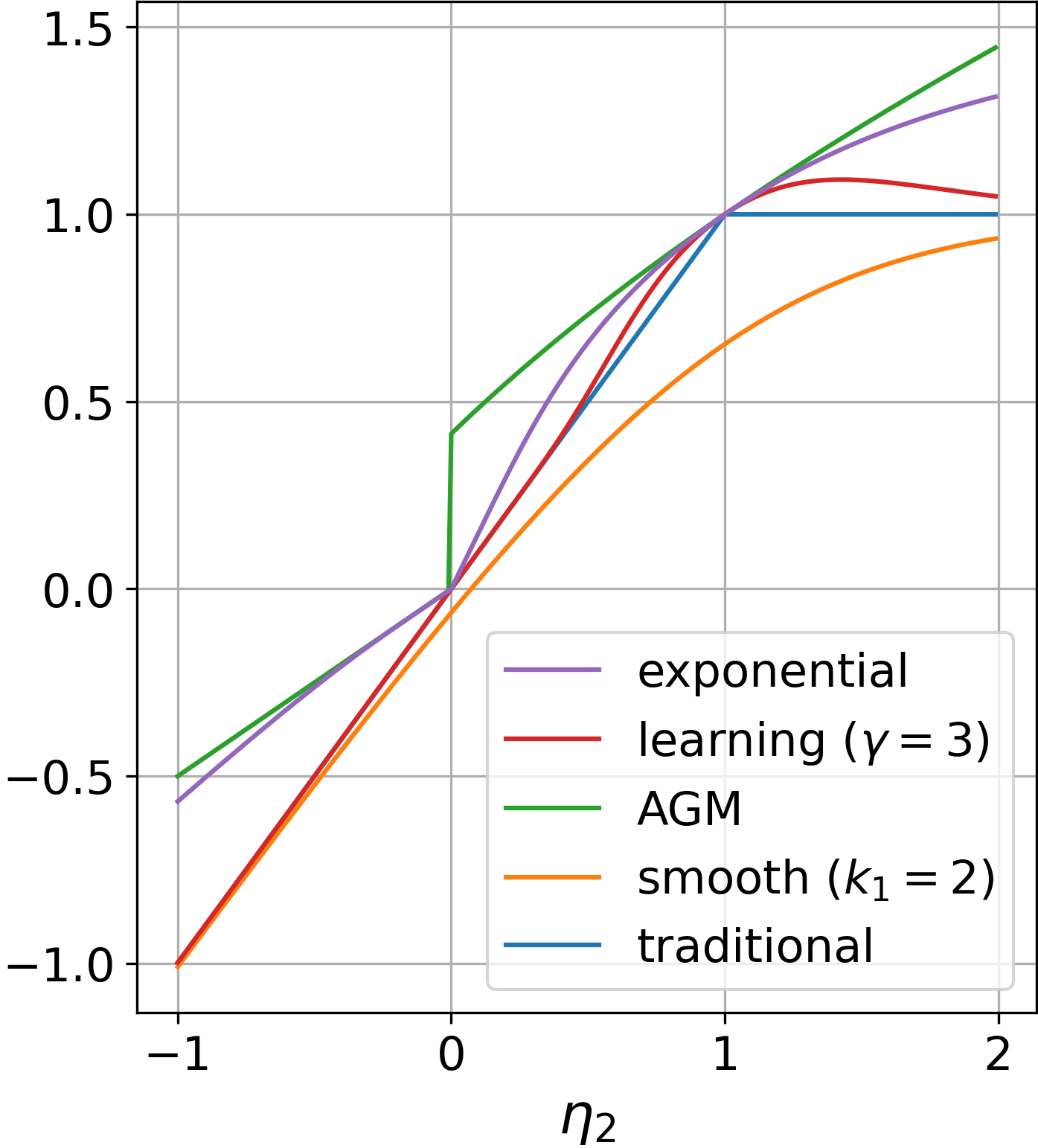}}
    \end{subfloat}
    \caption{\small Robustness of the conjunction $\mathcal A(\eta_1,\eta_2)$ for the traditional \cite{donze2010robust}, AGM \cite{mehdipour2019arithmetic}, learning \cite{varnai2020robustness}, smooth max-min \cite{gilpin2020smooth} and exponential ($\beta=0$) metrics as a function of $\eta_2$ with fixed values of $\eta_1$.\vspace{-10pt}} 
    \label{fig:compare}
    \vspace{-5pt}
\end{figure}

\vspace{-6pt}
\subsection{ Relationship between CaTL and CaTL+}

The syntax and qualitative semantics of CaTL are similar to the \emph{outer logic} of CaTL+, with two main differences. The first is the definition of a \emph{task}. In CaTL, a task is defined as a tuple $T^\prime=\langle d,\ \pi,\ \{c_i,m_i\}_{i\in I_T} \rangle$, where $d$ is a duration of time, $\pi$ is an atomic proposition specifying a region, $c_i$ is a capability and $m_i$ is a positive integer. A CaTL \emph{task} is satisfied if, between $[0,d]$, each of the regions labeled as $\pi$ contains at least $m_i$ agents with capability $c_i$ for all $\{c_i\}_{i\in I_T}$. There is no \emph{inner logic} in CaTL. Second, the ILP encoding requires that CaTL formulas contain no negation, while CaTL+ may contain negations as in \eqref{eq:syntax-in} and \eqref{eq:syntax-out}.
\begin{proposition}
With a given set of agents and the corresponding capabilities, specifications given by CaTL are a proper subset of specifications given by CaTL+.
\end{proposition}
\begin{proof}
Consider a CaTL \emph{task} $T^\prime=\langle d,\ \pi,\ \{c_i,m_i\}_{i\in I_T} \rangle$. Let $\varphi_\pi$ be a CaTL+ inner logic formula that represents the region specified by $\pi$. Then the CaTL \emph{task} $T^\prime$ can be transformed into a CaTL+ formula $\bigwedge_{i\in I_T}\mathbf G_{[0,d]}\langle \varphi,\ c_i,\ m_i \rangle$ which represents the equivalent specification. Since other operators of CaTL are included in CaTL+, any CaTL formula can be transformed into a CaTL+ formula with equivalent specification. On the other hand, a CaTL+ formula can be $\langle \mathbf F_{[a,b]} \varphi_\pi,\ c,\ m \rangle$, which cannot be specified by any CaTL formula due to its syntax. 
\end{proof}

Intuitively, a CaTL \emph{task} can only specify the number of agents that should ``always" exist in a region in a duration of time. All the other temporal and Boolean operators have to be outside the task. 
In contrast, CaTL+ is more expressive because a CaTL+ \emph{task} can specify a full STL formula in the inner logic. E.g., a CaTL+ task can be $T = \langle \mathbf F_{[a,b]} \varphi_\pi,\ c,\ m \rangle$. This \emph{task} is satisfied if $m$ agents satisfy $\varphi_\pi$ in $[a,b]$ synchronously or asynchronously. A CaTL formula can only specify $\mathbf F_{[a,b-1]}\langle 1,\pi,\{c,m\}\rangle$ which requires synchronous satisfaction. Note that by introducing new capability to each agent, the CaTL+ task $T = \langle \mathbf F_{[a,b]} \varphi_\pi,\ c,\ m \rangle$ can be transformed into an equivalent CaTL formula with combinatorially many conjunctions and disjunctions. However, the resulting CaTL formula will be very complex, which might make the control synthesis intractable. 

Moreover, CaTL+ formula can contain negations, which is absent from CaTL. So CaTL+ can specify more specifications including ``no more than $m$ agents could visit a region". 

The robustness definitions of CaTL+ and CaTL are also very different. CaTL robustness represents the minimum number of agents that can be removed from a given team in order to invalidate the given formula. It is not directly related to the trajectory of each agent. It is \emph{sound}, but without the \emph{continuity}, \emph{differentiability} and \emph{mask-eliminating} properties of CaTL+ (exponential) robustness.  

Another difference is that CaTL is defined on a discrete graph environment. The controls synthesized using ILP is a sequence of transitions between the vertices of the graph. In contrast, CaTL+ is defined on a continuous workspace. Each agent has its own dynamics and continuous control space. 

\section{Control Synthesis using CaTL+}
\label{sec:formulation}

In this section, we formulate and solve a CaTL+ control synthesis problem. In order to avoid unnecessary motions of the agents, we introduce a cost function over the controls. The overall optimization problem combines minimizing this cost with maximizing the CaTL+ (exponential) robustness. 

\begin{problem}
\label{pb:1}
Given a workspace $\mathcal S$, a set of agents $\{A_j\ |\ j\in \mathcal J \}$, a CaTL+ formula $\Phi$ with time horizon $H$, and a
weighted cost function $C(\cdot)\geq0$, find a control sequence $\mathbf u_j$ for each agent that maximizes the objective function:
\begin{equation}
    \label{eq:optimization}
    \begin{aligned}
    \max_{\mathbf u_j,j\in\mathcal J}\quad &\eta(\mathbf{S},\Phi,0) -  \frac {[\eta(\mathbf{S},\Phi,0)]_+} {\gamma}\cdot\sum_{j\in\mathcal J}C(\mathbf u_j) \\
    \text{s.t.}\quad & x_{j}(t+1) = f_j(x_{j}(t),u_{j}(t)),\\ &u_{j}(t)\in\mathcal U_j,\ t=0,\ldots,H-1,\\
    & l_j(x_j(t)) = s_j(t),\ j\in \mathcal J,\ t=0,\ldots,H,
    \end{aligned}
\end{equation}
where $\gamma$ is a parameter satisfying 
\begin{equation}
\label{eq:gamma}
    \gamma\geq\sup_{\mathbf u_j\in\mathcal U_j,j\in\mathcal J}\sum_{j\in\mathcal J}C(\mathbf u_j).
\end{equation}
\end{problem}

Due to the soundness of the robustness, $\eta(\mathbf S,\Phi,0)<0$ means that the CaTL+ formula $\Phi$ is not satisfied. In such situations, we focus on increasing the robustness without considering the cost $C$. When $\Phi$ is satisfied, i.e., $\eta(\mathbf S,\Phi,0)>0$, we try to minimize the cost at the same time. But minimizing the cost will never override the priority of satisfying the specification, because when \eqref{eq:gamma} is true, the cost cannot change the sign of the objective function. Hence, a positive objective ensures the satisfaction of the specification $\Phi$. 

Next, we propose a solution to Pb. \ref{pb:1}. The objective function in \eqref{eq:optimization} is highly non-convex, which means that there might exist many local optima. To avoid getting stuck at such points, we apply a two-step optimization: a global optimization followed by a local optimization. The result of the global optimization provides a good initialization for the local search, so the local optimizer is able to reach a point near the global optimum (for a highly non-convex function obtaining the exact global optimum is very difficult).  Specifically, for the global optimizer, we use Covariance Matrix
Adaptation Evaluation Strategy (CMA-ES) \cite{hansen2001completely}, which is a derivative-free, evalution-based optimization approach. Note that CMA-ES is primarily a local optimization approach, but it has also been reported to be reliable for global optimization with large population size \cite{hansen2004evaluating}. For the local optimizer, we compare two gradient based optimization methods, Sequential Quadratic Programming (SQP) and Limited-memory BFGS with box constraints (L-BFGS-B) \cite{bertsekas1997nonlinear}, in Sec. \ref{sec:simulation}.

An important issue of using gradient-based optimization is to compute the gradient efficiently. STLCG \cite{LeungArechigaEtAl2020} provides a way to compute the gradient of STL robustness. We adapt STLCG such that it also works for traditional and exponential robustness of CaTL+. Hence, we can obtain the gradient of the objective in \eqref{eq:optimization} automatically and analytically.

\section{Simulation Results}
\label{sec:simulation}
In this section we evaluate our algorithm on the earthquake emergency response scenario used as a running example throughout the paper. All algorithms were implemented in Python on a computer with 3.50GHz Core i7 CPU and 16GB RAM. For CMA-ES we used the pymoo \cite{pymoo} library. For both SQP and L-BFGS-B, we used the scipy package \cite{virtanen2020scipy}.

\subsection{Robustness-Based Optimization}
Consider the map shown in Fig. \ref{fig:map}. 
The control constraints $\mathcal U_g$ and $ \mathcal U_a$ for both ground and aerial vehicles are set to $[-1,1]\times[-1,1]$. We used $l^2$-norm as the cost function in \eqref{eq:optimization}, i.e., $C(\mathbf u_j)=\|\mathbf u_j\|_2$. Let $\gamma=1000$, satisfying \eqref{eq:gamma}.  

We first applied CMA-ES followed by SQP to solve Pb. \ref{pb:1}. Fig. \ref{fig:traj-exp} shows the resulting \emph{individual trajectories} for each agent. Agents' positions at some important time points are shown in Fig. \ref{fig:result}. The resulting \emph{team trajectory} satisfies the specification $\Phi$. Note that the agents do not need to stay in a region at the same time to satisfy a \emph{task} like $\Phi_1=\langle \mathbf F_{[0,8]} s\in C,\ ``Delivery",\ 6 \rangle$ or $\Phi_2=\langle \mathbf F_{[0,25]} s\in V_1,\ ``Delivery",\ 3 \rangle \land \langle \mathbf F_{[0,25]} s\in V_2,\ ``Delivery",\ 3 \rangle$. In fact, on the premise of satisfying $\Phi_1$, the two aerial vehicles depart from region $C$ to inspect the bridge before all the ground vehicles reach region $C$. Moreover, both villages $V_1$ and $V_2$ are visited by $4$ agents though the requirement is $3$, which makes the exponential robustness higher.

We compare the result with using traditional robustness in Pb. \ref{pb:1}. After the same number of iterations as the exponential robustness, the resulting trajectories using traditional robustness are shown in Fig.\ref{fig:traj-trad}. The specification is also satisfied, but the agents have many unnecessary motions, and they do not arrive at the centers of the regions. In addition, $V_1$ and $V_2$ are visited by exactly $3$ agents. There are two reasons: (1) the objective function with traditional robustness is hard to optimize; (2) traditional robustness only considers the most critical points, so many behaviors are masked.  When using exponential robustness, more agents reach the centers of the regions, and each center is visited for more time points. 

\begin{figure}
    \centering
    \begin{subfloat}[\label{fig:traj-exp}]{
    \includegraphics[width=3.6cm]{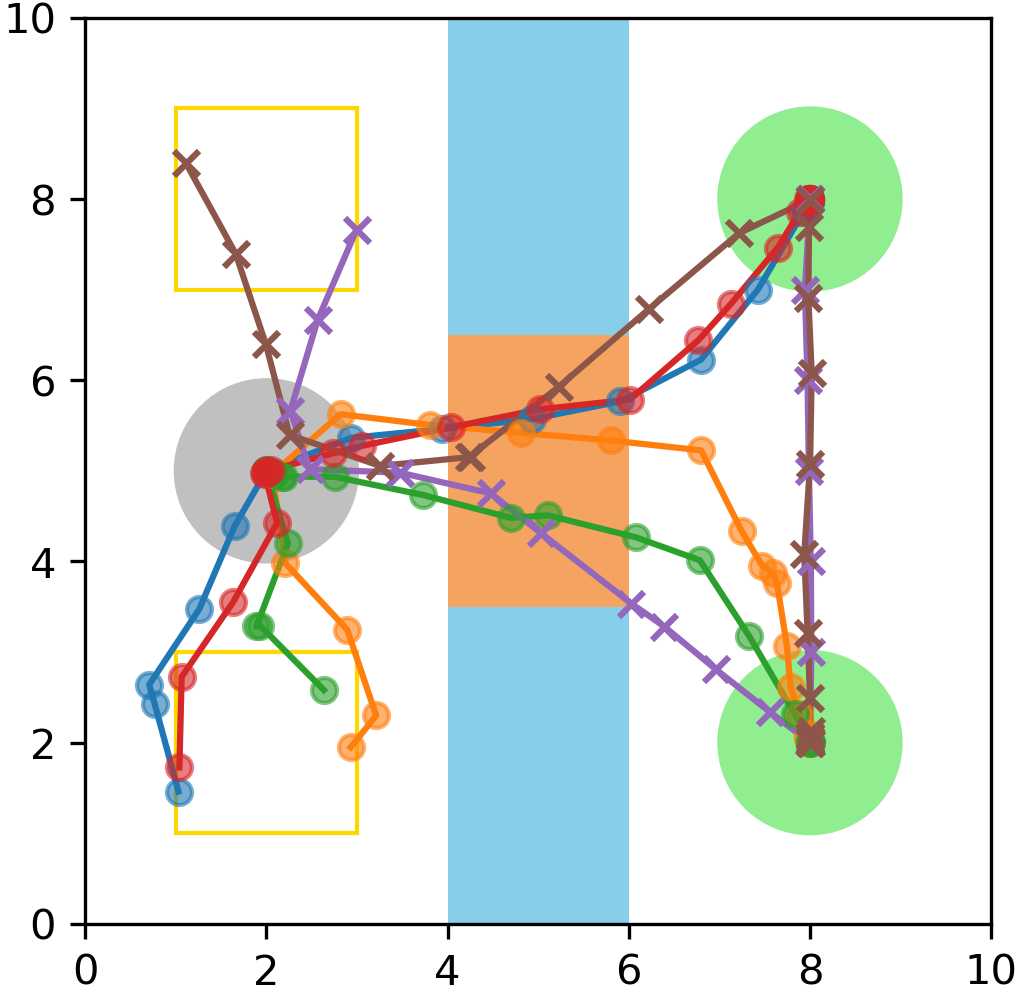}}
    \end{subfloat}
    \quad
    \begin{subfloat}[\label{fig:traj-trad}]{
    \includegraphics[width=3.6cm]{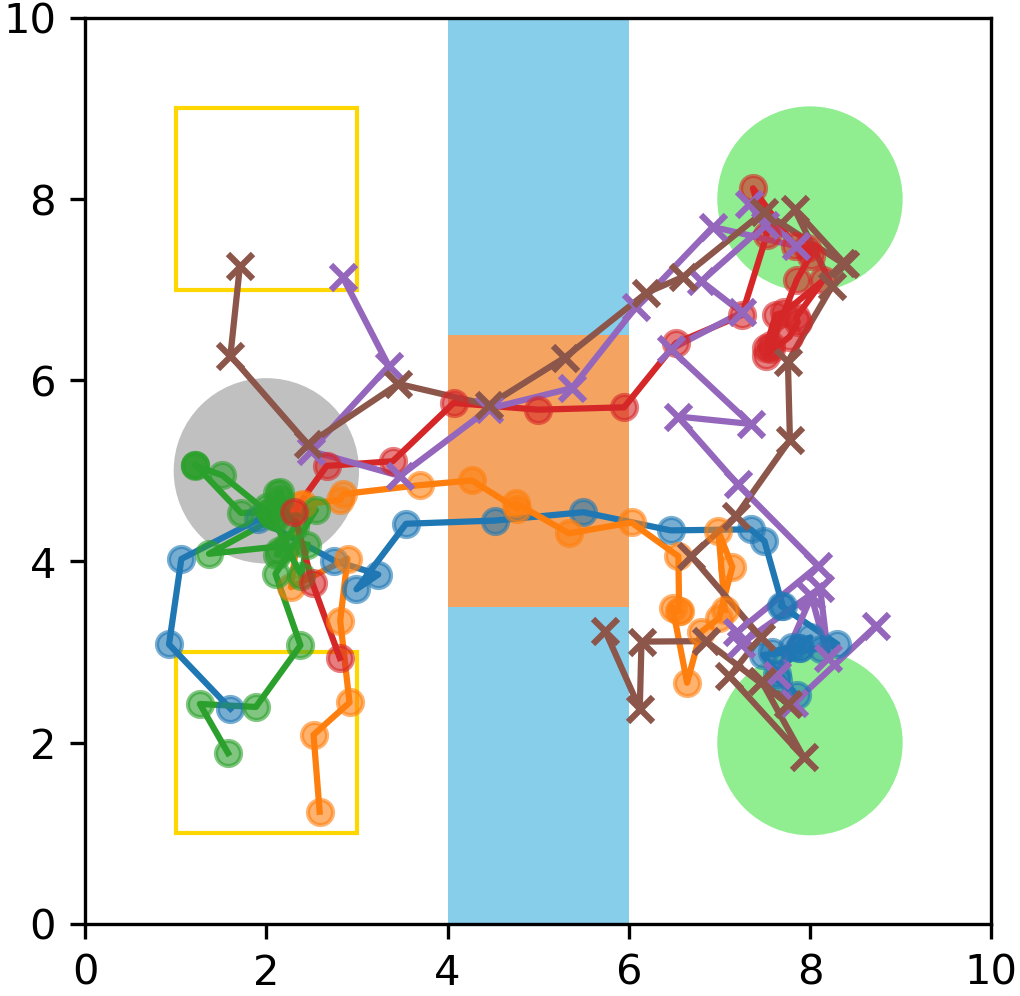}}
    \end{subfloat}
    \caption{\small (a) Trajectories obtained solving Pb. \ref{pb:1}. (b) Trajectories obtained solving Pb. \ref{pb:1} with the traditional robustness. Brown and purple trajectories are for the aerial vehicles, the other trajectories are for the ground vehicles. 
    Each agent's positions between adjacent time points are connected by straight lines.\vspace{-10pt}}
    \label{fig:traj}
    \vspace{-5pt}
\end{figure}

\begin{figure*}
     \centering
     \begin{subfloat}[$t=5$\label{fig:time5}]{
         \centering
         \includegraphics[width=0.185\textwidth]{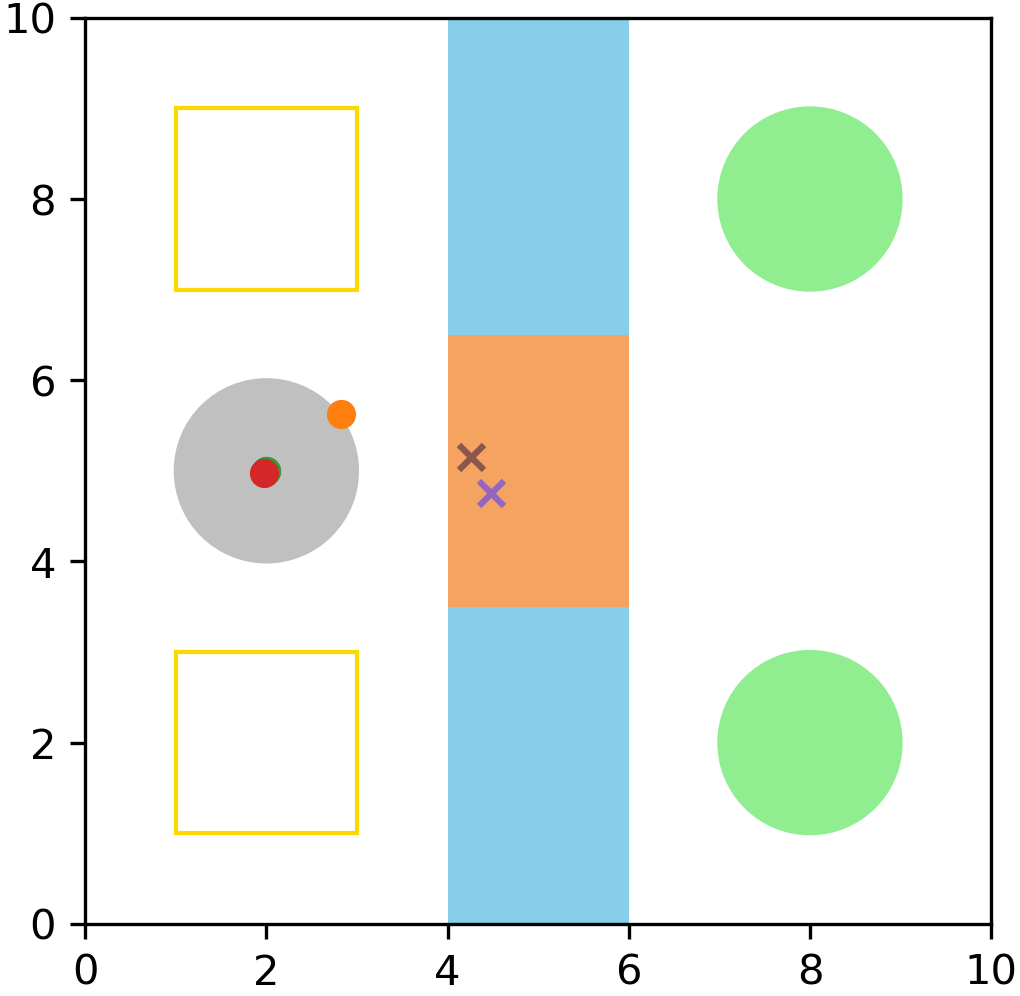}}
     \end{subfloat}
     \begin{subfloat}[$t=9$\label{fig:time7}]{
         \centering
         \includegraphics[width=0.185\textwidth]{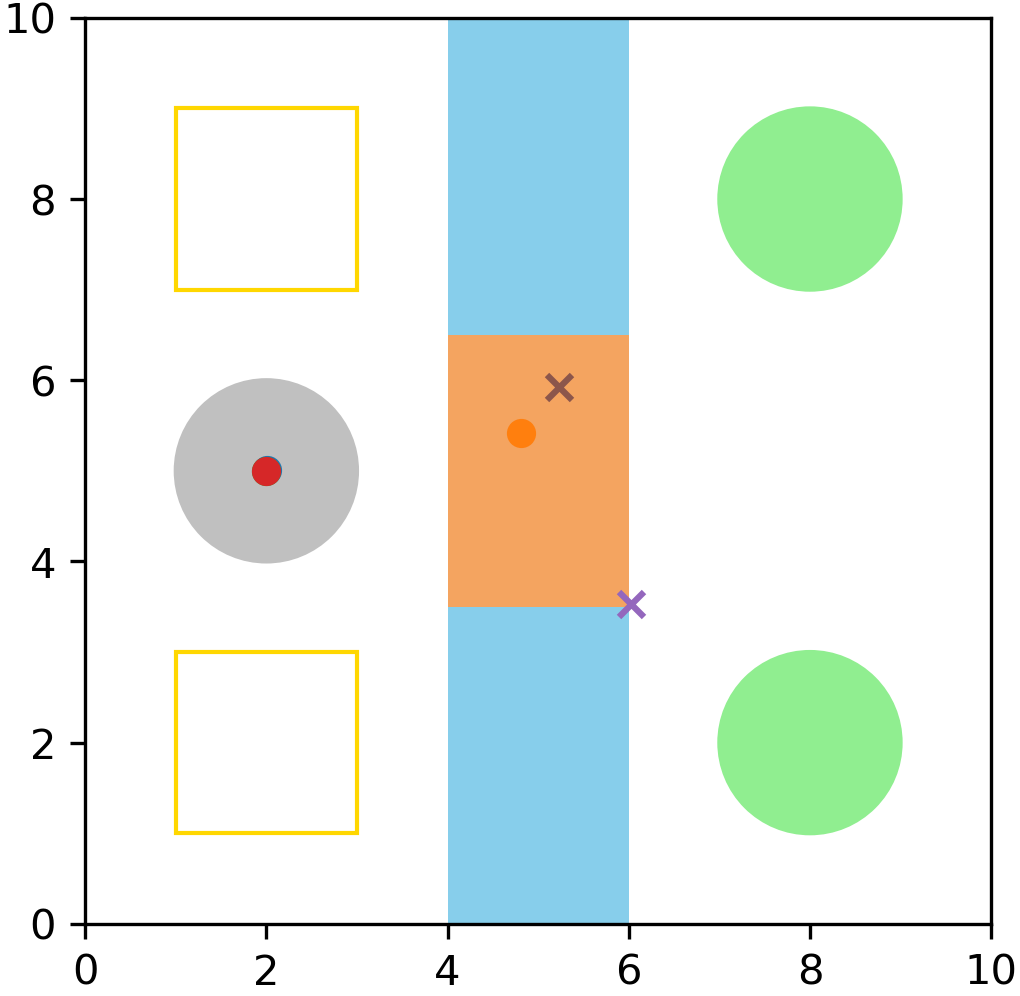}}
     \end{subfloat}
    \begin{subfloat}[$t=11$\label{fig:time10}]{
         \centering
         \includegraphics[width=0.185\textwidth]{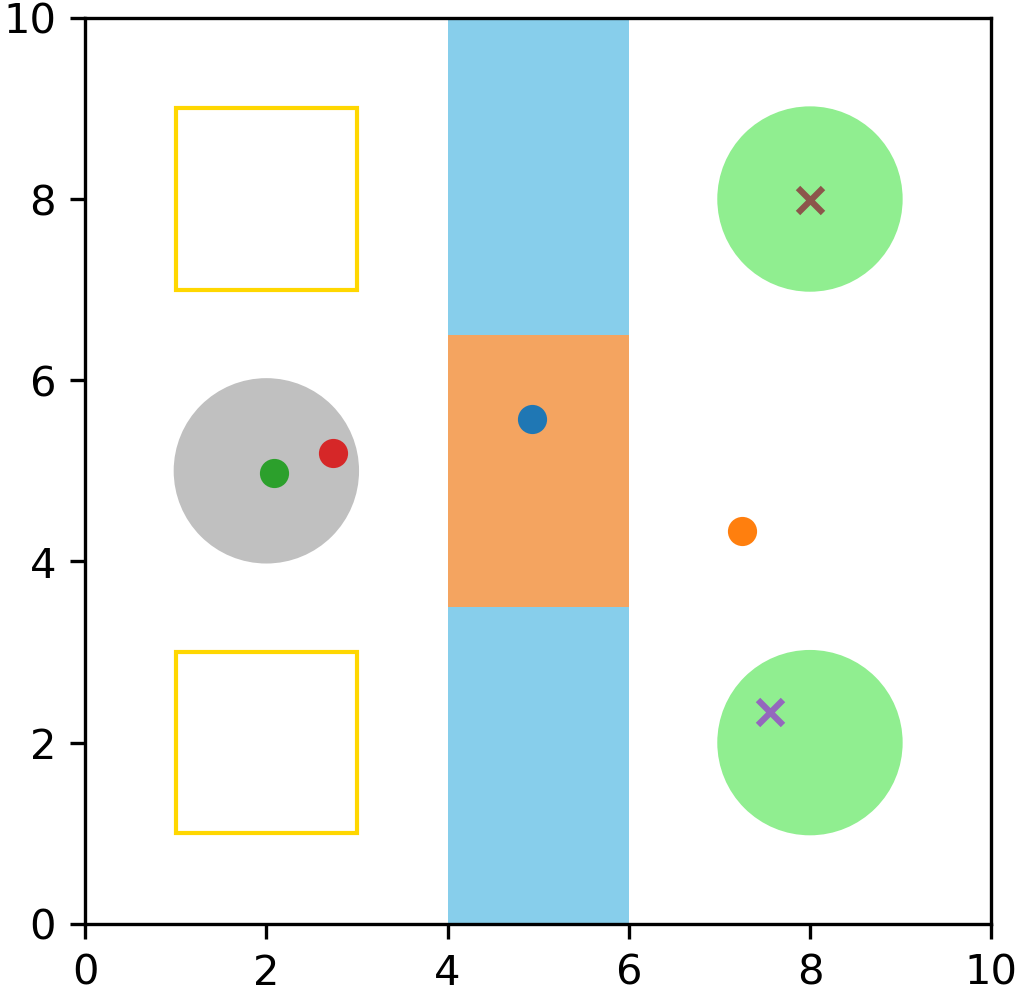}}
     \end{subfloat}
     \begin{subfloat}[$t=13$\label{fig:time13}]{
         \centering
         \includegraphics[width=0.185\textwidth]{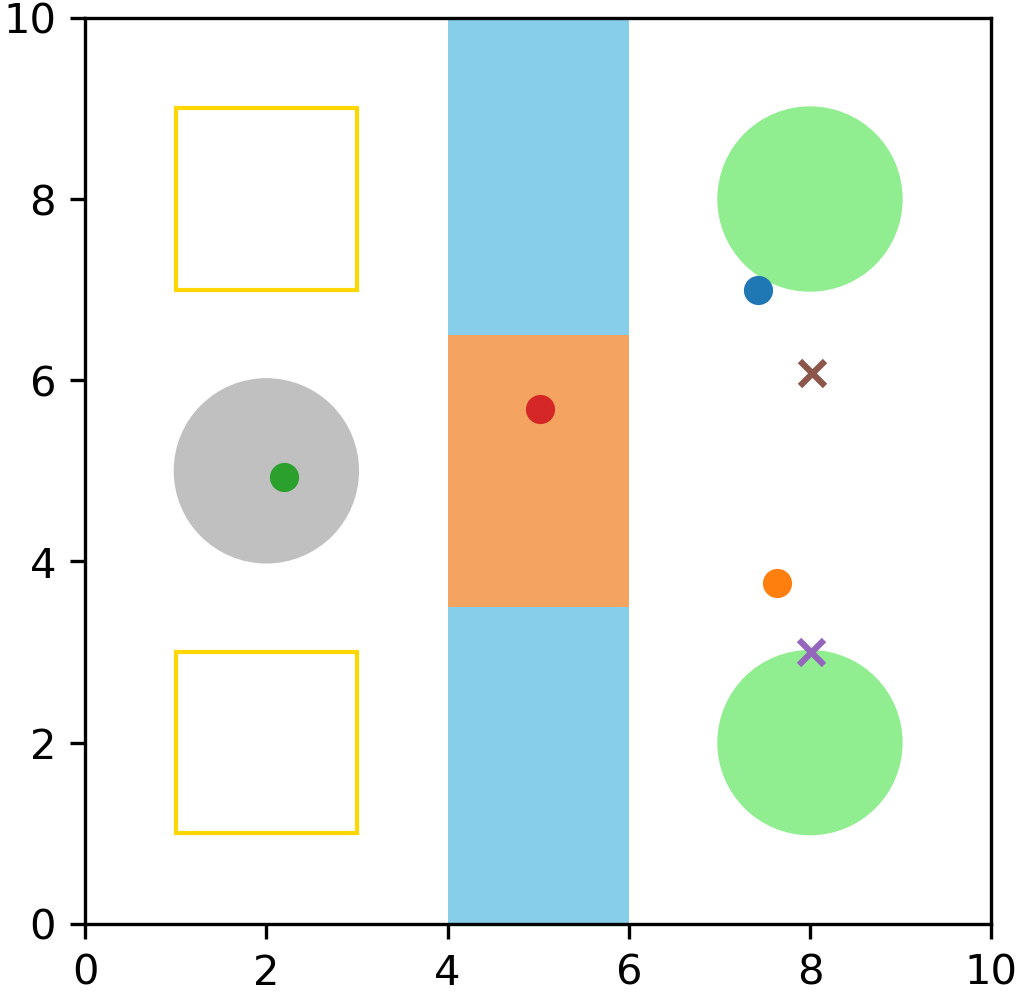}}
     \end{subfloat}
     \begin{subfloat}[$t=17$\label{fig:time17}]{
         \centering
         \includegraphics[width=0.19\textwidth]{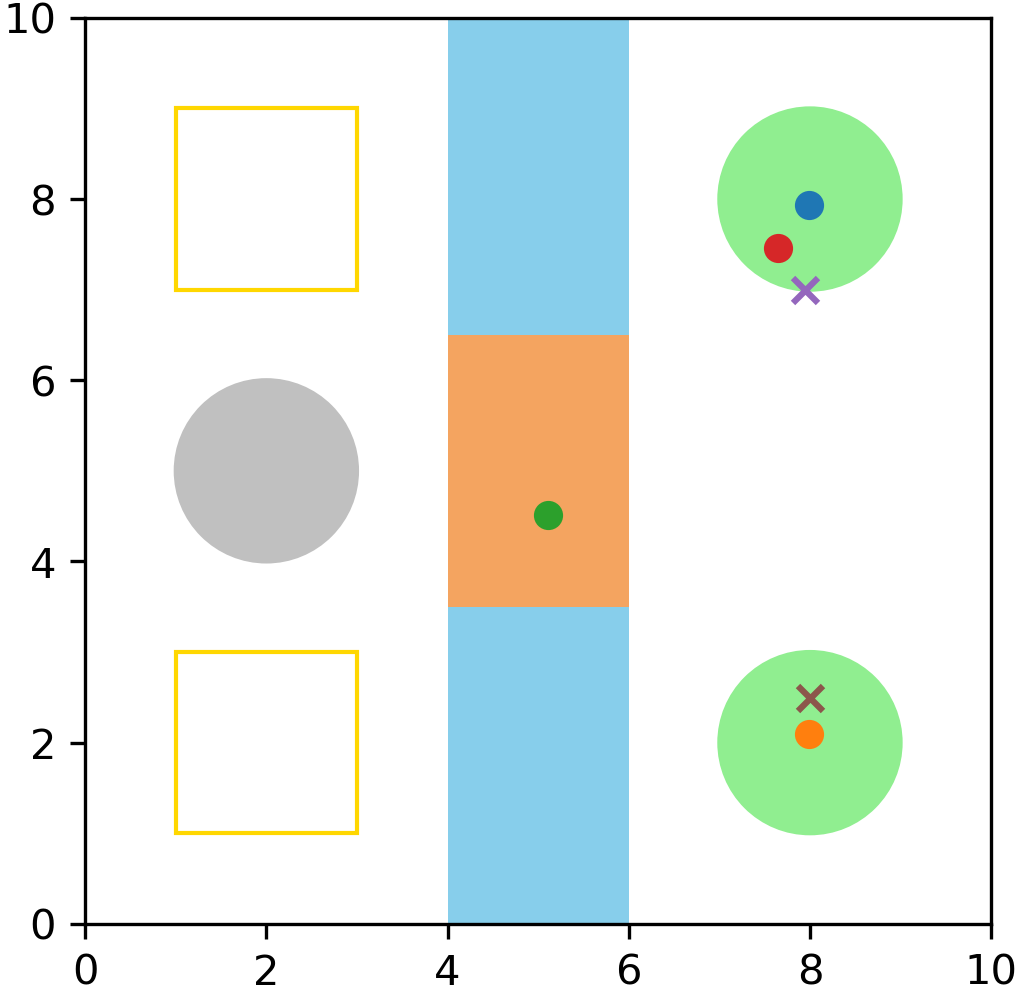}}
     \end{subfloat}
    \caption{\small Demonstration of the team trajectory at selected time points. Circular points represent ground vehicles and ``x" points represent aerial vehicles. Fig. \ref{fig:time5}: at time $t=5$ all ground vehicles arrive at region $C$, while the two aerial vehicles inspect the bridge $B$. Fig. \ref{fig:time7}: at time $t=7$ the first ground vehicle crosses the bridge and the the two aerial vehicles go to $V_1$ and $V_2$, respectively. Fig. \ref{fig:time10}-\ref{fig:time17}: at time $t=10,13,17$ the second, third and fourth ground vehicles cross the bridge one by one, and the two aerial vehicles switch there positions.\vspace{-5pt}}
\label{fig:result}
\vspace{-5pt}
\end{figure*}

\subsection{Scalability}
Next we investigate the scalability of our algorithm. We proportionally increase the number of agents in each category until the total number is $21$. The number of optimization variables, i.e., the number of controls, grows linearly with respect to the number of agents. At the same time, $\Phi_2$, $\Phi_3$, $\Phi_5$ remain unchanged, while $\Phi_1$, $\Phi_4$, $\Phi_6$, which specify requirements for all agents (in a category), change accordingly. 

We also compare the results of using SQP and L-BFGS-B in the local optimization phase. For each different number of agents and local optimizer, we solve Pb. \ref{pb:1} for $10$ times over different initial conditions. Each time we iterate CMA-ES for fixed $1000$ generations with population size $150$. For SQP and L-BFGS-B, we set the maximum number of iterations to $200$ and the local optimizer terminates after convergence. Then we compute the mean and standard deviation of the resulting robustness and computation time. The results are shown in Fig. \ref{fig:scale}. In Fig. \ref{fig:time}, we can see that the computation time for the global search (CMA-ES) increases almost linearly with respect to the number of agents. The number of agents has a large effect on the computation time of SQP, but a relatively small effect for L-BFGS-B. We note that the number of agents does not change the time to compute the exponential robustness and its gradients. Hence, the different effects on the computation time come from the optimizers themselves.  SQP is not designed for high dimensional optimization, so large number of variables make SQP inefficient. Computation other than evaluating the robustness and its gradients is trival in L-BFGS-B, hence the computation time is basically the same with increasing number of agents. In Fig. \ref{fig:robust}, the positive robustness indicates that Pb. \ref{pb:1} can be solved successfully when considering more agents. 
Although L-BFGS-B is faster and has better scalability, SQP returns higher robustness. Moreover, SQP is able to deal with nonlinear constraints, while L-BFGS-B is restricted to box constraints as in this example.

\begin{figure}
\vspace{-5pt}
     \centering
     \begin{subfloat}[\label{fig:time}]{
         \includegraphics[width=3.8cm]{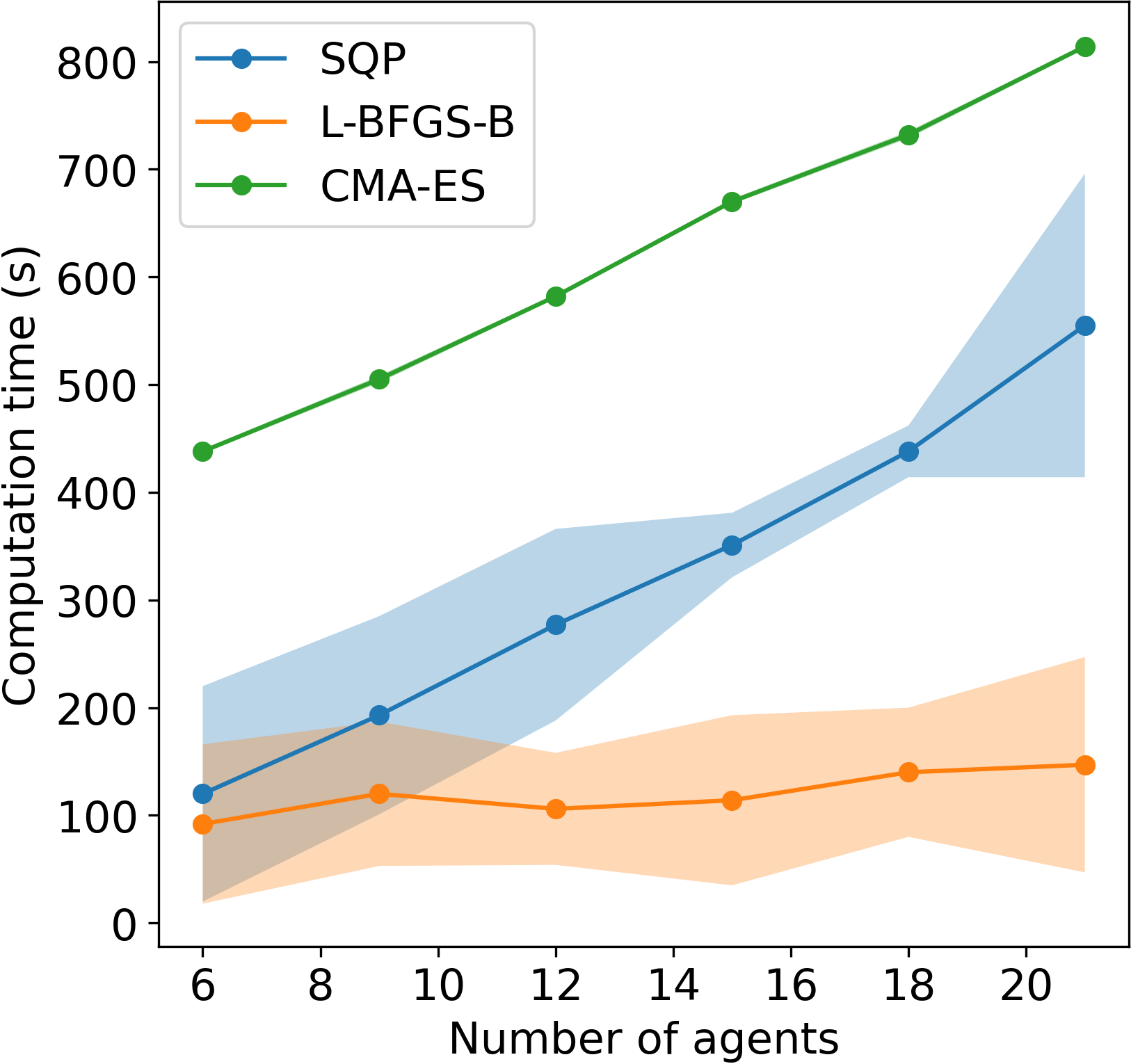}}
     \end{subfloat}
     \ 
     \begin{subfloat}[\label{fig:robust}]{
         \includegraphics[width=3.8cm]{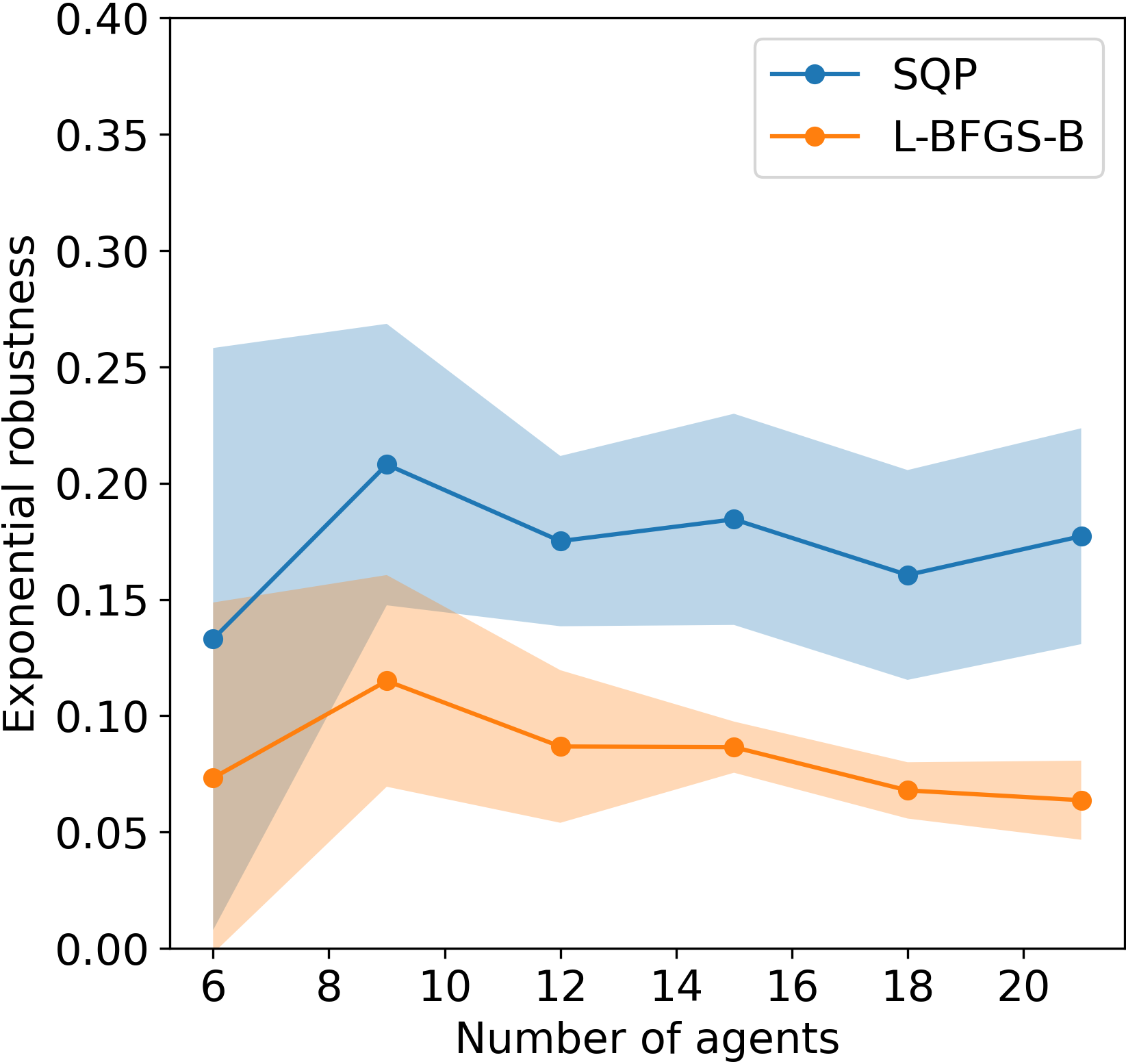}}
     \end{subfloat}
    \caption{\small Computation results with increasing number of agents. Solid lines indicate the mean value. Shadowed areas indicate the standard deviation. (a) shows the computation time for the global and local optimizers. (b) shows the resulting exponential robustness values.\vspace{-10pt}}
\label{fig:scale}
\vspace{-5pt}
\end{figure}

In \cite{leahy2021scalable}, control synthesis from CaTL is solved using ILP. The computation time grows exponentially with respect to the environment size, i.e., the number of discrete states in the environment. Hence, it is intractable on finely segmented environment. In our approach, the environment is a continuous space and our approach can deal with regions other than rectangles, such as circles. However, the SQP solver is inefficient for large number of variables, which limits the total number of agents we can control, while the L-BFGS-B solver has good scalability but gives less robust results. Other solvers can be investigated in the future work.

\section{Conclusion and Future Work}
\label{sec:conclusion}
We introduced a new logic called CaTL+, which is convenient to specify temporal logic requirements for multi-agent systems over continuous workspace and is strictly more expressive than CaTL. We defined two quantitative semantics for CaTL+: the traditional and the exponential robustness. The latter is sound, differentiable almost everywhere and has the mask-eliminating property. A two-step optimization strategy was proposed for control synthesis from CaTL+ formula. The simulation results illustrate the efficacy of our approach. One limitation is that we do not consider inter-agent collisions and the behavior of the system between adjacent time points. In future work, we will investigate incorporating a lower level controller that ensures the correct dense-time behavior and collision avoidance. We also plan to employ more efficient optimizers.






\bibliographystyle{IEEEtran}
\bibliography{references}

\end{document}